\documentclass[%
  reprint,            %
superscriptaddress,
 aps,
floatfix,
]{revtex4-2}

\makeatletter
\renewcommand{\@makefnmark}{\hbox{\normalfont[\@thefnmark]}}
\renewcommand{\@makefntext}[1]{%
  \noindent[\@thefnmark]\enspace #1%
}
\makeatother

\usepackage[normalem]{ulem}
\usepackage{graphicx}%
\usepackage{dcolumn}%
\usepackage{booktabs}%
\usepackage{bm}%
\usepackage{braket}
\usepackage[dvipsnames]{xcolor}
\usepackage{amsfonts}
\usepackage{amsmath,amssymb}
\usepackage{mathtools}
\usepackage{xfrac} %
\usepackage{hyperref}%
\usepackage{cleveref}
\crefname{equation}{Eq.}{Eqs.}
\Crefname{equation}{Eq.}{Eqs.}
\crefname{figure}{Fig.}{Figs.}
\Crefname{figure}{Fig.}{Figs.}
\crefname{section}{Sec.}{Secs.}
\Crefname{section}{Sec.}{Secs.}
\crefname{subsection}{Sec.}{Secs.}
\Crefname{subsection}{Sec.}{Secs.}
\crefname{subsubsection}{Sec.}{Secs.}
\Crefname{subsubsection}{Sec.}{Secs.}
\crefname{appendix}{Appendix}{Appendices}
\Crefname{appendix}{Appendix}{Appendices}
\usepackage{comment}

\DeclareMathOperator{\tr}{tr}

\usepackage{amsthm}
\newtheorem{theorem}{Theorem}
\newtheorem{lemma}{Lemma}

\crefname{protocol}{Protocol}{Protocols}
\Crefname{protocol}{Protocol}{Protocols}

\crefname{procedure}{Procedure}{Procedures}
\Crefname{procedure}{Procedure}{Procedures}
\crefname{theorem}{Theorem}{Theorems}
\Crefname{theorem}{Theorem}{Theorems}
\crefname{lemma}{Lemma}{Lemmas}
\Crefname{lemma}{Lemma}{Lemmas}

\newcommand{\im}{\mathrm{i}}
\newcommand{\e}{\mathrm{e}}
\newcommand{\diff}{\mathrm{d}}

\usepackage{xparse}
\NewDocumentCommand{\PT}{g}{\mathrm{PTER}\IfValueT{#1}{_{#1}}}

\newcommand{\fsc}[1]{}

\begin{document}

\title{Unbiased Hamiltonian Simulation by Reversing Trotter Error Dynamics}

\author{Keisuke Murota}
\affiliation{Quantinuum K.K., Otemachi Financial City Grand Cube 3F, 1-9-2 Otemachi, Chiyoda-ku, Tokyo, Japan}
\affiliation{Department of Physics, The University of Tokyo, Tokyo 113-0033, Japan}

\author{Yuta Kikuchi}
\affiliation{Quantinuum K.K., Otemachi Financial City Grand Cube 3F, 1-9-2 Otemachi, Chiyoda-ku, Tokyo, Japan}
\affiliation{RIKEN Center for Interdisciplinary Theoretical and Mathematical Sciences (iTHEMS), RIKEN, Wako 351-0198, Japan}

\author{Enrico Rinaldi}
\affiliation{Quantinuum, Partnership House, Carlisle Place, London SW1P 1BX, United Kingdom}

\author{Fr\'ed\'eric Sauvage}
\affiliation{Quantinuum, Partnership House, Carlisle Place, London SW1P 1BX, United Kingdom}

\author{Synge Todo}
\affiliation{Department of Physics, The University of Tokyo, Tokyo 113-0033, Japan}
\affiliation{Institute for Physics of Intelligence, The University of Tokyo, Tokyo 113-0033, Japan}
\affiliation{Institute for Solid State Physics, The University of Tokyo, Kashiwa, 277-8581, Japan}

\begin{abstract}
Owing to their simplicity and low overhead, Suzuki--Trotter formulas remain the de facto Hamiltonian simulation methods on current quantum computing platforms.
Systematic Trotter errors, however, will quickly become limiting when scaling to larger problems and aiming for higher accuracy.
We present a mechanism that removes the systematic error of any $k$-th order Suzuki--Trotter simulation, at the cost of a constant sampling overhead. The key insight is that the Trotter error is itself a coherent dynamics to be reversed, rather than a deviation to be bounded.
By identifying the structure of this error in closed form, we carry out that reversal through quasi-probabilistic decompositions.
The resulting algorithm, called Probabilistic Trotter Error Reversal (PTER),
is unbiased and still improves the gate-count scaling
of Suzuki--Trotter formulas while retaining their simplicity.
Numerical simulations of a Heisenberg spin chain support the predicted resource advantage already at modest system sizes.
\end{abstract}

\maketitle

\paragraph*{Introduction. ---}

Predicting the real-time dynamics of strongly correlated quantum matter, from spin liquids to lattice gauge theories and chemical reactions, is a central goal of physics and chemistry~\cite{Feynman1982, Lloyd1996, SavaryBalents2017, Banuls2020, McArdle2020}. Classical methods often struggle to reach the physically relevant scales~\cite{Baez2020}.
Early fault-tolerant quantum computers are expected to bring this goal within reach.
The algorithmic target is a time-evolution circuit that is simple enough for realistic hardware and whose gate-count scaling is close to the near-optimal $\widetilde{\Omega}(nt)$ scaling set by the Lieb--Robinson bound~\cite{LiebRobinson1972, HaahHastingsKothariLow2021, ChildsSu2019} for an $n$-qubit lattice Hamiltonian and a simulation time $t$.
Suzuki--Trotter formulas~\cite{Trotter1959, Suzuki1976} already deliver the hardware simplicity: their locality preservation, low overhead, and absence of ancilla registers have made them the standard building block for quantum simulation~\cite{Suzuki1991, Childs2021trotter, Childs2018}.
This simplicity, however, comes with suboptimal scalings that limit their scalability:
a $k$-th order formula requires $\mathcal{O}(n^{1+1/k}\,t^{1+1/k}\,\epsilon^{-1/k})$ gates to guarantee a systematic error (i.e., algorithmic bias) $\epsilon$~\cite{Childs2021trotter}.
This has motivated the search for algorithms going beyond Trotter formulas.

Quantum singular-value transformation (QSVT)~\cite{Low2017, Gilyen2019, Low2019} almost saturates the lower bound on query complexity~\cite{BerryAhokasCleveSanders2007, BerryChildsKothari2015}, at the cost of block encodings, ancillas, and controlled access to Hamiltonian terms, all too demanding for near-future hardware.
Randomised circuit sampling improves error scalings in various ways~\cite{Campbell:2019fez, Childs2019, Chen2021, Nakaji:2023gze, Pocrnic:2023lrz, Zhao2022RandomInputs, Faehrmann2022}, and can sometimes completely eliminate the algorithmic bias, but does not leverage the geometric locality of Trotter formulas~\cite{Granet2023,Chakraborty2024, kiumi2025te, murota2026errormitigated, Hayata2026}.
Hybrid schemes compensate the Trotter error with randomised corrections, through linear combinations of unitaries~\cite{zeng2025simple} or by sampling the error unitary~\cite{cho2024doubling, mineh2025steer}.
They improve the gate-count scaling but still retain an algorithmic bias that is difficult to bound tightly, and can require controlled rotations and a larger prefactor.
To our knowledge, no existing method delivers improved gate-count scaling, unbiasedness, and Trotter simplicity simultaneously, and this work closes that gap.

Our starting observation is that, for a small Trotter step $\tau$, the error of a $k$-th order Trotter formula is generated by a time-dependent \emph{remainder Hamiltonian} $G_k(s)$ satisfying $\|G_k(s)\|=\mathcal{O}(s^k)$ and admitting a nested-commutator closed form.
Earlier analyses used this commutator structure to tighten error bounds~\cite{Childs2021trotter}.
Instead, we use it to %
eliminate the Trotter error by reversing the dynamics generated by $G_k$ after each Trotter step (see~\cref{fig:overview}).
We implement the reverse dynamics through TE-PAI (time-evolution via probabilistic angle interpolation)~\cite{kiumi2025te,koczor2024probabilistic}.
Being based on a quasi-probabilistic decomposition, it involves random sampling of circuits and additional classical post-processing (see inset of \cref{fig:overview}).
The latter increases the number of measurements %
required by a multiplicative factor, called \emph{sampling overhead}.
In what follows, we only consider settings in which such overhead is \emph{constant}, although greater than one.

The resulting scheme, a Probabilistic Trotter Error Reversal at order $k$ ($\PT_k$), converts systematic Trotter errors into constant sampling overhead.
For $n$-qubit geometrically local Hamiltonians, with each Pauli term supported on $\mathcal{O}(1)$ neighbouring sites of a fixed-dimensional lattice, $\PT{k}$ has expected gate count
$\mathcal{O}(n^{1+1/(2k+1)}\,t^{1+1/(2k+1)})$, improving on the corresponding $k$-th order Trotter formula and rapidly approaching the lower bound $\widetilde{\Omega}(nt)$ when increasing $k$.
For general Hamiltonians, sums of $N$ Pauli terms with $1$-norm $\beta$, $\PT{k}$ requires $\mathcal{O}(N\beta\,t^{1+1/(2k+1)})$ gates, improving the time and norm dependence of the corresponding Trotter formulas (see \cref{app:trotter-recap}).
As an additional benefit, since $G_k$ is expressed in closed form, both the sampling of the circuits and the evaluation of their expected gate count are precise and classically efficient at any system size.
We use this to compare gate counts for the simulation of a one-dimensional Heisenberg chain, confirming that $\PT_2$ requires substantially fewer gates than the corresponding Suzuki--Trotter formula, and outperforms it even at small system sizes.
By leveraging the reversibility of the Trotter error dynamics,
 these results provide a route to larger-system and higher-accuracy many-body simulations without sacrificing the simplicity of Trotter formulas.

In what follows, we implement the evolution $U(t)=\e^{-\im Ht}$, for an $n$-qubit Hamiltonian $H$ which is a weighted sum of $N$ Pauli strings $\{P_j\}_{j=1}^N$, as a circuit.
Throughout, the gate count is taken to be the number of Pauli rotations $\e^{-\im\theta P_j}$ in the circuit, with $\theta$ an arbitrary angle. This is the unit of circuit complexity used in all theorems, figures, and bounds below. 
We start by exposing the working details of $\PT$ at order $k=1$ in a simplified setting, before generalising.

\paragraph*{Simplified setting. ---}
Consider a Hamiltonian with splitting $H=A+B$, with $A$ and $B$ consisting of mutually commuting Pauli strings. We further take $H$ to be geometrically local, such that $N=\mathcal{O}(n)$.

\begin{figure}[!t]
\centering
    \includegraphics[width=\columnwidth]{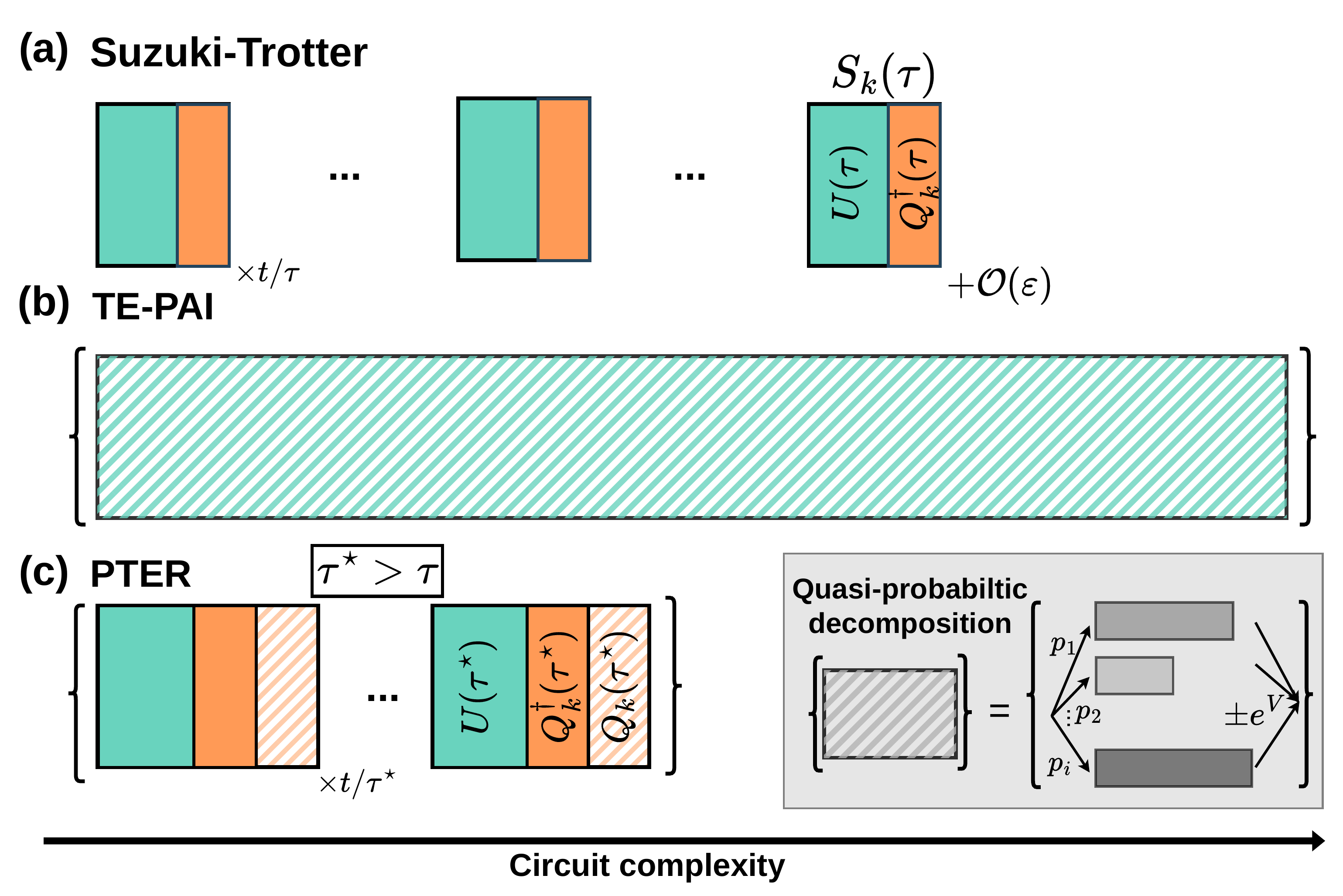}
    \caption{
    Simulation of $U(t)=\e^{-\im H t}$. (a) A $k$-th order Trotter formula $\mathcal{S}_k(\tau)$ incurs a unitary Trotter error $Q_k^\dagger(\tau)$ at each step of length $\tau$. (b) TE-PAI relies on quasi-probabilistic decomposition, whereby circuits are randomly sampled and recombined through post-processing (inset).
    When used to simulate $U(t)$, it is unbiased and requires neither additional qubits nor controlled gates, but incurs a larger expected gate count.
    (c) Through a suitable characterisation of the time-dependent generator $G_k(s)$ of $Q^\dag_k$, we reverse the error dynamics at each step.
    The resulting $\PT{k}$ construction removes any Trotter bias without ancillas, at the cost of a \emph{constant} sampling overhead. This enables us to increase the step length $\tau^\star$, and ultimately to reduce the expected gate counts.
    }
    \label{fig:overview}
\end{figure}

The evolution is split into $r$ steps of length $\tau=t/r$, so that $U(t)=U(\tau)^r$.
The first-order Trotter formula
$\mathcal{S}_1(\tau)=\e^{-\im A\tau}\e^{-\im B\tau}$
approximates $U(\tau)$ with an error stemming from the non-vanishing commutator $[A,B]$.
To account for this error, define the \emph{correction unitary} $Q_1(\tau) := U(\tau)\,\mathcal{S}_1^\dag(\tau)$.
The desired evolution $U(\tau)=Q_1(\tau)\mathcal{S}_1(\tau)$ can then be implemented by applying the Trotter formula followed by the correction unitary.

To analyse this correction,
consider its adjoint $Q_1^\dag(s)=\mathcal{S}_1(s)U^\dag(s)$,
which is the Trotter error in the form of a unitary operator.
Differentiating this error yields
$\im \partial_s Q_1^\dag(s) = G_1(s)\,Q_1^\dag(s)$,
where $G_1(s):=\im(\partial_s\mathcal{S}_1(s))\mathcal{S}_1^\dag(s)-\mathcal{S}_1(s)H\mathcal{S}_1^\dag(s)$ is the time-dependent remainder Hamiltonian,
which coherently generates the Trotter error dynamics.
As derived in \cref{app:G1-kernel}, it can be expressed as
\begin{equation}
  G_1(s) = -\int_0^s \e^{-\im As}\,\e^{-\im B\sigma}\,\im[A,B]\,\e^{\im B\sigma}\,\e^{\im As} \,\diff\sigma.
  \label{eq:G1-def}
\end{equation}
We use this form to implement the correction channel
$\mathcal{Q}_1(\tau)[\cdot]:=Q_1(\tau)[\cdot]Q_1^\dagger(\tau)$
through TE-PAI (\cref{fig:overview}b)~\footnote{In fact, we detail the circuit implementation of $Q_1(\tau)^\dagger[\cdot]Q_1(\tau)$. It relates to the form quoted in the main text simply by taking the adjoint of the circuits.}, 
which realises unbiased Hamiltonian simulation (as a channel) through quasi-probabilistic decomposition (see \cref{app:tepai-details}).
As it involves sampling unitary circuits and classical post-processing,
its resources are captured by (i) an expected gate count, over the sampled unitaries~\footnote{All the gate counts quoted in this work are expected counts $\mathbb{E}[N_{\mathrm{g}}]$. The gate counts of the individual sampled circuits concentrate around this expected value with a standard deviation $\sqrt{\mathbb{E}[N_{\mathrm{g}}]}$ according to the Poisson distribution.}, and (ii) a multiplicative sampling overhead.
TE-PAI applies to Hamiltonians decomposed into Hermitian-involution (h.i.) operators $R$ with $R^2=\mathbb{I}$ and $R^\dagger=R$.
For a decomposition $G=\sum_j c_jR_j$ in a generic h.i. basis, let $\|G\|_{1,\mathrm{hi}}:=\sum_j|c_j|$, and denote it $\|G\|_{1,\mathrm{Pauli}}$ for a decomposition in the Pauli basis. We recall the resource counts of TE-PAI below (see \cref{app:tepai-resources}).

\begin{lemma}[TE-PAI~\cite{kiumi2025te}]\label{lem:tepai}
Let $G(s)=\sum_j c_j(s)\,R_j$ be a time-dependent Hamiltonian with h.i. decomposition and define the integrated rate
$\lambda:=\int_0^\tau \|G(s)\|_{1,\mathrm{hi}}\,\diff s$.
For any $V>0$, TE-PAI
implements unbiased Hamiltonian simulation of
$G(s)$ for a time $\tau$
with sampling overhead $\e^{V}$, and expected gate count
\begin{equation}
  \mathbb{E}[N_{\mathrm{g}}] = \frac{4\lambda^2}{V} + \frac{V}{2}.
  \label{eq:tepai-cost}
\end{equation}
For $r$ independent $\tau$-time simulations of $G(s)$ within the same circuit, to maintain a total sampling overhead $\e^V$,
replace $V\rightarrow V / r$ in~\cref{eq:tepai-cost} to get the expected gate count per simulation.%
\end{lemma}

At fixed overhead $\e^V$, employing TE-PAI directly for the simulation of $H$ yields $\lambda \in \mathcal{O}(t)$ and an expected number of gates $\mathbb{E}[N_{\mathrm{g}}] \in \mathcal{O}(t^2)$, quadratic in the simulation time, worse than any Trotter formula.
In $\PT{1}$, however, TE-PAI is only applied to reverse the errors generated by $G_1$.
To assess the gate count of this reversal, first expand $[A,B]=\im \sum_h c_hP_h$ in the Pauli basis, and substitute it into \cref{eq:G1-def} to get
$G_1(s) = \sum_h c_h \int_0^s \e^{-\im As}\,\e^{-\im B\sigma}\,P_h\,\e^{\im B\sigma}\,\e^{\im As}\,\diff\sigma$.
Further decomposing $G_1(s)$ in the Pauli basis and noting that the integrands satisfy $\e^{-\im As}\,\e^{-\im B\sigma}\,P_h\,\e^{\im B\sigma}\,\e^{\im As}=P_h+\mathcal{O}(s)$, one obtains $\|G_1(s)\|_{1,\mathrm{Pauli}}=\alpha_1\,s+\mathcal{O}(s^2)$ with $\alpha_1:=\sum_h|c_h|$.

The integrated rate then evaluates to
$\lambda_1 := \int_0^\tau\!\|G_1(s)\|_{1,\mathrm{Pauli}}\,\diff s
= \alpha_1\tau^2/2+\mathcal{O}(\tau^3)$.
For a geometrically local $H$, $\alpha_1=\mathcal{O}(n)$ since each Pauli in $[A,B]$ is supported on $\mathcal{O}(1)$ sites.
The gate count per step of $\PT{1}$ combines the $N$ gates needed to implement the Trotter formula with the expected $\mathcal{O}(\lambda_1^2/V)$ gates needed for the correction.
This is captured by the following theorem (proven in \cref{app:gate-count-unified}).

\begin{theorem}\label{thm:gate-count}
Let $H=A+B$ be an $n$-qubit geometrically local Hamiltonian with $N = \mathcal{O}(n)$ Pauli terms, and $V>0$.
Taking a Trotter step $\tau^\star=\mathcal{O}(V^{1/3}N^{1/3}\alpha_1^{-2/3}t^{-1/3})$,
$\PT{1}$ implements 
unbiased Hamiltonian simulation 
of $H$ for a time $t$ with sampling overhead $\e^V$ and expected gate count
\begin{equation}
  \mathbb{E}[N_{\mathrm{g}}]
  = A_1\,V^{-1/3}\,N^{2/3}\,\alpha_1^{2/3}\,t^{4/3}\left( 1 + o(1)\right),
  \label{eq:gates-optimal}
\end{equation}
with $A_1=3/2^{2/3}$.
For constant $V$,
$\mathbb{E}[N_{\mathrm{g}}]=\mathcal{O}(n^{4/3}t^{4/3})$.
\end{theorem}
\begin{proof}[Proof sketch]
At fixed $r$, the expected total gate count is
$\mathbb{E}[N_{\mathrm{g}}] = rN + \alpha_1^2t^4/(Vr^2) + V/2$.
The first term accounts for the Trotter circuits, the remaining ones come from the corrections implemented by TE-PAI as per \cref{eq:tepai-cost} with $\lambda_1$ evaluated above.
It is minimised at $r^\star=(2\alpha_1^2t^4/(VN))^{1/3}$, resulting in \cref{eq:gates-optimal}.
\end{proof}

\paragraph*{Geometrically local Hamiltonians. ---}
The previous treatment extends to $\PT{k}$, for $k=1$ and even $k\geq2$,
and to general splittings $H=\sum_{\ell=1}^L H_\ell$.
As before, the order-$k$ Trotter error, defined as $Q_k(s):=U(s)\mathcal{S}_k^\dagger(s)$, is a coherent dynamics generated by a time-dependent remainder Hamiltonian $G_k(s)$ with $\lVert G_k(s)\rVert=\mathcal{O}(s^k)$.
Let us express a general $k$-th order Trotter formula in the form
$\mathcal{S}_k(\tau)=\prod_{j=1}^{\tilde\Upsilon_k}\e^{-\im\tau \widetilde H_j}$,

where $\widetilde H_j := \omega_j H_{\ell_j}$, with  Suzuki coefficients $\omega_j$ and indices $\ell_j\in\{1,\dots,L\}$ both fixed by the formula used.
This product has $\tilde\Upsilon_k=\Upsilon_k L$ layers, with $\Upsilon_k$ being the number of \emph{stages}~\cite{Childs2021trotter} of the formula ($\Upsilon_1=1$ and $\Upsilon_k=2\times5^{k/2-1}$ for even $k\ge2$).
To proceed, we express $G_k$ in a form
amenable to simulation through~\cref{lem:tepai}.

\begin{lemma}[Kernel form of $G_k$]\label{lem:Gk-kernel-main}
For a $k$-th order Trotter formula $\mathcal{S}_k$ of $H=\sum_\ell H_\ell$, the remainder Hamiltonian is
\begin{equation}
  G_k(s)
  = \sum_\nu\int_0^{s} K_\nu(s,\sigma)\,L_\nu C_\nu L_\nu^\dagger\,\diff\sigma,
  \label{eq:Gk-kernel-main}
\end{equation}
where each $C_\nu$ is a $k$-fold nested commutator of the components $\{H_\ell\}$, each $L_\nu(s,\sigma)$ is a partial product of Trotter factors with $L_\nu=\mathbb{I}+\mathcal{O}(s)$, and the non-negative kernels obey $\int_0^s K_\nu(s,\sigma)\,\diff\sigma = s^k w_\nu$. Hence $\|G_k(s)\|=\mathcal{O}(s^k)$. The explicit $L_\nu$, $K_\nu$, $w_\nu$ and a proof are given in \cref{app:higher-order}.
\end{lemma}
From~\cref{lem:Gk-kernel-main}, we see that the leading term of $G_k(s)$ is $s^k\sum_\nu w_\nu C_\nu$.
It follows that the remainder Hamiltonian has norm $\|G_k(s)\|_{1,\mathrm{Pauli}} = \alpha_k s^k(1+\mathcal{O}(s))$, and that $\lambda_k = \alpha_k\,\tau^{k+1}(1+\mathcal{O}(\tau))/(k+1)$,
where we have defined
$\alpha_k := \bigl\|\sum_\nu w_\nu\,C_\nu\bigr\|_{1,\mathrm{Pauli}}$.
We stress 
that $\alpha_k$ captures cancellations between the Pauli terms of the different nested commutators $C_{\nu}$ and thus can be much smaller than the standard commutator norm $\tilde\alpha_k := \sum_\nu w_\nu\,\|C_\nu\|_{1,\mathrm{Pauli}}$ commonly used to bound Trotter errors~\cite{Childs2021trotter}.
As for \cref{thm:gate-count}, the Trotter step is chosen to minimise the overall gate count and we get
(proven in \cref{app:gate-count-unified}):

\begin{theorem}
\label{thm:gate-count-k}
Let $k=1$ or $k\geq 2$ even, $H=\sum_\ell H_\ell$ be an $n$-qubit geometrically local Hamiltonian
with $N=\mathcal{O}(n)$ Pauli terms, and
$V>0$.
By taking a Trotter step
$\tau^\star=\mathcal{O}(V^{\frac{1}{2k+1}}(\Upsilon_k N)^{\frac{1}{2k+1}}\alpha_k^{-\frac{2}{2k+1}}t^{-\frac{1}{2k+1}})$,
$\PT{k}$ implements unbiased
Hamiltonian simulation of
$H$ for a time $t$ with sampling overhead $\e^V$ and expected gate count
\begin{equation}
\begin{split}
  \mathbb{E}[N_{\mathrm{g}}]
  = &A_k\,V^{-\frac{1}{2k+1}}\,\\  & \times (\Upsilon_k N)^{1-\frac{1}{2k+1}}
  \alpha_k^{\frac{2}{2k+1}}\,t^{1+\frac{1}{2k+1}}\bigl(1+o(1)\bigr),
  \label{eq:gates-optimal-k}
\end{split}
\end{equation}
with $A_k \leq 2 $.
For constant $V$,
$\mathbb{E}[N_{\mathrm{g}}]=\mathcal{O}(n^{1+1/(2k+1)}t^{1+1/(2k+1)})$.
\end{theorem}
Note that the expected gate count in \cref{eq:gates-optimal-k} carries an explicit small prefactor $A_k$, so the advantage of $\PT$ holds at the level of constants, not only in scaling.
We further stress that both the sampling of the circuits and the evaluation of the expected gate counts are exact and require only $\mathcal{O}(N)$ classical pre-computation (\cref{app:sampling}). This is possible as our analysis of the Trotter error avoids resorting to Baker--Campbell--Hausdorff (BCH) expansions~\cite{CasasMurua2009}, whose number of terms grows exponentially with the order and which therefore require truncations, but relies on the closed form of the remainder Hamiltonian.

\paragraph*{General (non-local) Hamiltonians. ---}
For completeness, we briefly summarise results for a general Hamiltonian $H$, with $N$ Pauli terms and $1$-norm $\beta := \|H\|_{1,\mathrm{Pauli}}$.
In essence, the treatment of the non-local case follows the steps of the geometrically local one, with the exception that we now need to use $\tilde{\alpha}_k$, rather than $\alpha_k$, when evaluating $\lambda_k$, due to implementation subtleties (see \cref{app:nongeom-sampler}).
For a sampling overhead $\e^V$ we obtain (see \cref{app:non-geometric-gate-count}):
\begin{equation}
  \mathbb{E}[N_{\mathrm{g}}]
  = \mathcal{O}\!\left(V^{-\frac{1}{2k+1}}\,N\,\beta\,t^{1+\frac{1}{2k+1}}\right).
  \label{eq:gates-nonlocal}
\end{equation}
At constant sampling overhead, this reduces to $\mathcal{O}(N\beta\,t^{1+1/(2k+1)})$ which improves over the $k$-th order Trotter scaling $\mathcal{O}(N\beta^{1+1/k}\,t^{1+1/k}\,\epsilon^{-1/k})$ in every parameter $\epsilon$, $t$, and $\beta$.

\begin{figure*}[!htp]
  \centering
  \includegraphics[width=\textwidth]{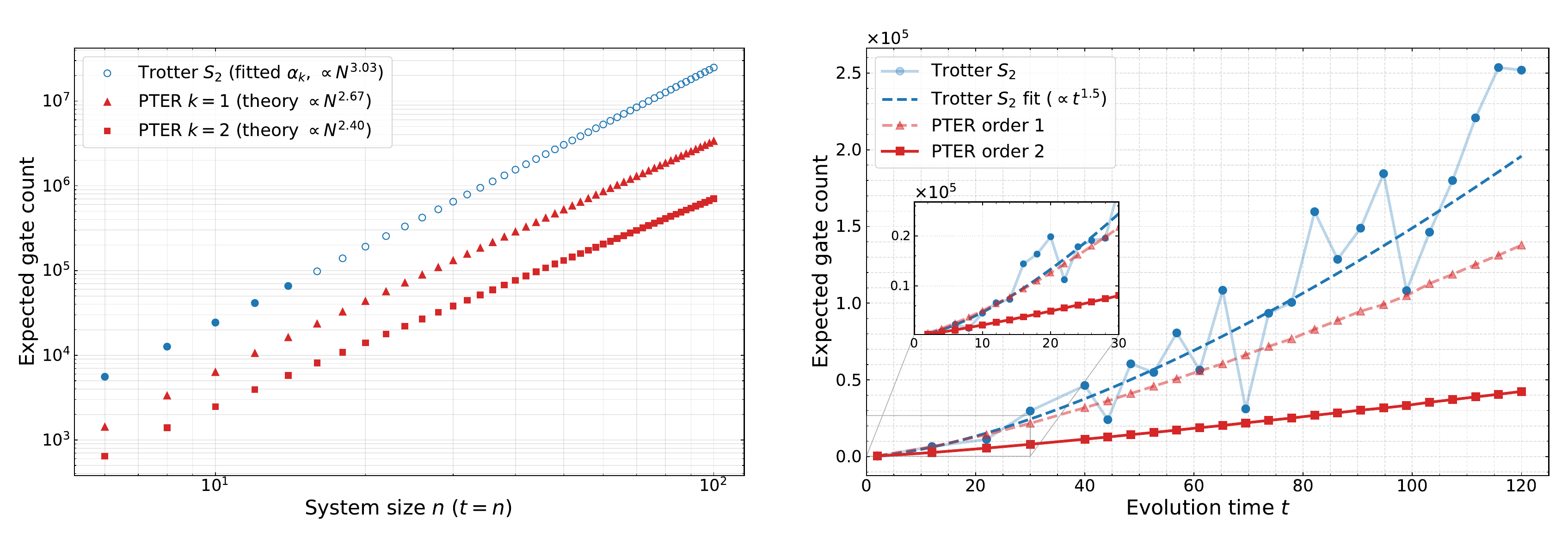}
  \caption{
    Expected gate count $\mathbb{E}[N_{\mathrm{g}}]$ to guarantee an RMSE~\eqref{eq:rmse} below $\epsilon_{\rm t}=0.03$, given $M=10^4$ measurements, for the simulation of a one-dimensional Heisenberg chain.
    Left: for arbitrary initial states and observables, as a function of the system size $n$ (fixing the time $t=n$). Order-$2$ Trotter circuits (blue dots, with empty markers to highlight extrapolated data) are compared to $\PT_k$ circuits for $k=1$ (red triangles) and $k=2$ (red squares).
    Right: for a fixed initial state (N\'eel state) and a fixed single-qubit observable ($Z_1$), as a function of $t$ (fixing $n=10$). The order-$2$ Trotter gate counts (blue dots) are fitted with a power law (dashed line).
  }
  \label{fig:numerics}
\end{figure*}

\paragraph*{Numerical results. ---}
Compared with Trotter formulas, $\PT$ achieves improved gate-count scalings in both the system size and the simulation time while remaining unbiased.
To appreciate further its practical reach, we conduct numerical analyses for the simulation of a one-dimensional Heisenberg chain with periodic boundary conditions ($n+1 \equiv1$). 
Explicitly, $H=\sum_{j=1}^{n}h_j$ with $h_j:=X_jX_{j+1}+Y_jY_{j+1}+Z_jZ_{j+1}$, that is recast as $H=A+B$ with $A/B=\sum_{j\in{\mathrm{odd/even}}}h_j$.
In many problems of interest~\cite{Fauseweh2024, Miessen2023}, one evolves an initial state and estimates the expectation value of an observable $O$ (with $\|O\|_{\infty}\leq 1$) through $M$ repeated measurements. The accuracy of such estimates is quantified by the root-mean-square error,
\begin{equation}
\mathrm{RMSE}=\sqrt{\mathrm{bias}^2+\Delta_{\rm 1s}^2/M},
  \label{eq:rmse}
\end{equation}
that captures the systematic error (bias) and the statistical fluctuation (variance) given by the single-shot variance $\Delta_{\rm 1s}^2$ rescaled by the number of measurements~\cite{murota2026errormitigated}.
The following studies compare the expected gate count to achieve a target RMSE $\epsilon_{\rm t}=0.03$ given $M=10^4$ measurements.
For PTER circuits, the bias is zero (independent of the state or observable) while $\Delta^2_{\rm 1s}\leq \e^V$ accounting for the sampling overhead. This fixes $V$ through
$\e^{V/2}=3$ for the parameters used here, such that the expected gate counts are computed directly from~\cref{thm:gate-count,thm:gate-count-k}.
For Trotter circuits, $\Delta^2_{\rm 1s}\leq 1$ but the bias will depend on the specific choice of initial states and observables.

We first examine the resources required to guarantee an RMSE below $\epsilon_{\rm t}$, for arbitrary initial state and observable, as a function of the system size $n$ (fixing $t=n$).
For Trotter circuits, we bound the bias by the diamond-norm distance between the Trotterized and exact evolution channels, which is the worst-case bias over all pure input states and Hermitian observables (see \cref{app:error-metric}).
Thus, the step count $r^*$ is chosen as the smallest $r$ for which the resulting RMSE bound is below $\epsilon_{\rm t}$.
Since evaluating the diamond-norm distance requires exact diagonalisation, the resulting gate counts are evaluated for $n\le14$ and extrapolated beyond this by fitting $r^*=a n^b$.
The resulting gate counts (left panel of \cref{fig:numerics}) are consistent with the predicted separation between the second-order Trotter baseline (blue dots), scaling as $\mathcal{O}(n^3)$, 
and $\PT_k$, scaling as $\mathcal{O}(n^{2+2/3})$ for $k=1$ (red triangles) and $\mathcal{O}(n^{2.4})$ for $k=2$ (red squares).
As seen, despite the overhead incurred by a more complex underlying Trotter formula, a larger $\Upsilon_k$ in~\cref{eq:gates-optimal-k}, $\PT_2$ improves on $\PT_1$.
More importantly, this shows that even at small system sizes ($n \sim 10$), it already improves on the Trotter formula by almost an order of magnitude.
Given the respective scalings, this gap increases to a $30$-fold reduction in the gate count at system size $n=100$.

The biases incurred by Trotter formulas can be substantially smaller than the diamond-norm bound for a specific initial state
and observable, especially if the latter is local~\cite{Heyl2019, Layden2022}.
To further challenge $\PT$, we restrict the comparison to a fixed initial state (the N\'eel state $\ket{01\cdots01}$) and a fixed local observable ($Z_1$), at $n=10$ (right panel of \cref{fig:numerics}).
For the Trotter circuits, we thus compute~\cref{eq:rmse} with the bias evaluated for this specific state and observable, rather than with the worst-case bound used earlier.
The gate counts reported correspond again to the smallest counts achieving $\mathrm{RMSE}\le\epsilon_{\rm t}$. 
For $\mathcal{S}_2$ (blue dots) the gate count fluctuates in time but overall follows the expected $t^{3/2}$ scaling (dashed curve), to be contrasted with $\PT_{k}$, which scales as $t^{4/3}$ (for $k=1$) or $t^{6/5}$ (for $k=2$).
At large times this hierarchy of exponents is visible.
More strikingly, even at short time (inset), ${\rm PTER}_1$ matches the second-order Trotter gate count
while ${\rm PTER}_2$ systematically improves it.
This proof-of-principle benchmark supports the claim that, due to its improved scalings and %
small prefactors, PTER can already bring benefits to tasks addressable with current quantum computing resources, running a few thousand gates.

\paragraph*{Discussion and outlook. ---}

In this work, we have re-interpreted the Trotter error not as an accuracy barrier to be reduced by shortening the step size, but as a coherent dynamics, characterised in closed form free from the complications of BCH expansions, which can be reversed.
After reversal, the Trotter step is rather increased, yielding the improved gate-count scalings of \cref{eq:gates-optimal-k,eq:gates-nonlocal}, at constant sampling overhead, while retaining the hardware simplicity of bare Trotter circuits.
Additional notable benefits of $\PT$ are worth highlighting. 
First, as estimate errors become purely statistical (free of bias), precise resource estimates are possible, in contrast to Trotter formulas, whose estimates rest on bias bounds that typically overestimate the actual error by orders of magnitude~\cite{Childs2018}. This is valuable for algorithms such as ground-state energy estimation via quantum phase estimation~\cite{Yoshioka2024, Kanasugi2026qpe, KimmelLowYoder2015, Russo2021rpe}. Second, all sampled correction gates are rotations by a fixed and tuneable angle, a feature inherited from TE-PAI. Such regularity can readily be exploited in Clifford+T implementations, for example through catalyst towers, promising substantial T-gate savings~\cite{Bothe2026smallangle,SunKoczor2025rotations}.

The methods presented here lend themselves to many extensions and follow-up studies.
While the reverse evolution is realised through TE-PAI, other unbiased time-dependent simulation methods~\cite{Granet2023, Nakaji:2023gze, Chakraborty2024} could be substituted.
Furthermore, generalising $\PT$ to time-dependent Hamiltonian simulations and to Lindblad evolutions would be of great interest.
Finally, it is valuable to assess $\PT$ on more demanding simulation tasks, notably two-dimensional lattice systems where quantum simulation is expected to offer an advantage over classical methods~\cite{Yoshioka2024}, and to benchmark it against multi-product formulas or Trotter-step extrapolation schemes~\cite{Vazquez2023, WatsonWatkins2024}, and QSVT-based algorithms~\cite{Low2017, Gilyen2019, Low2019}.

The intended regime for $\PT$ is early fault-tolerant quantum computing~\cite{LinTong2022, Dutkiewicz2024}, where gate errors are reduced but block-encoding infrastructure remains too demanding.
Its improved scalings, however, suggest that its edge %
is likely to extend beyond the early fault-tolerant regime.
For geometrically local Hamiltonians, the best known lattice algorithm achieves the near-optimal $\mathcal{O}(nt\,\mathrm{polylog}(nt/\epsilon))$ gate count via a Lieb--Robinson decomposition~\cite{HaahHastingsKothariLow2021}, while $\PT{k}$ achieves $\mathcal{O}((nt)^{1+1/(2k+1)})$.
Even under the conservative assumptions of comparable prefactors and a polylogarithmic factor reduced to $\log(nt)$, $\PT{2}$ remains favourable up to $nt\sim 10^{6}$ and $\PT{4}$ up to $nt\sim 10^{12}$.
Overall, $\PT$ holds the promise of advancing Hamiltonian simulation to larger systems and higher accuracy, in the near term and beyond.

\begin{acknowledgments}
We thank \'Etienne Granet and Matthias Rosenkranz for their feedback on the manuscript.
This work was supported by the Center of Innovation for Sustainable Quantum AI (SQAI), JST Grant Number JPMJPF2221, and JSPS KAKENHI Grant Numbers JP24K00543 and JP24KJ0892.
\end{acknowledgments}

\bibliography{main}

\clearpage
\appendix
\onecolumngrid

\setcounter{theorem}{0}\setcounter{lemma}{0}\setcounter{proposition}{0}
\setcounter{corollary}{0}\setcounter{definition}{0}\setcounter{example}{0}
\setcounter{remark}{0}\setcounter{protocol}{0}\setcounter{procedure}{0}
\renewcommand{\thetheorem}{S\arabic{theorem}}
\renewcommand{\thelemma}{S\arabic{lemma}}
\renewcommand{\theproposition}{S\arabic{proposition}}
\renewcommand{\thecorollary}{S\arabic{corollary}}
\renewcommand{\thedefinition}{S\arabic{definition}}
\renewcommand{\theexample}{S\arabic{example}}
\renewcommand{\theremark}{S\arabic{remark}}
\renewcommand{\theprotocol}{S\arabic{protocol}}
\renewcommand{\theprocedure}{S\arabic{procedure}}

The appendices are organised in five parts.
\Cref{app:trotter_error} derives a nested-commutator representation of the remainder Hamiltonian $G_k(s)$, beginning with its definition and building up from first order ($k=1$) through general $k$.
\Cref{app:tepai-details} reviews TE-PAI~\cite{kiumi2025te}, which is the framework used to implement the correction channels without any bias. We further detail concepts of quasi-probabilistic decompositions and sampling overhead, and proves~Lemma~1 of the main text. 
\Cref{app:apply-tepai} combines the two previous parts by applying TE-PAI for the simulation of $G_k$. There, we establish the gate count per step of $\PT$, and the optimisation of this gate count to prove Theorems~1 and~2 of the main text 
and their extensions to general (non-local) Hamiltonians.
\Cref{app:sampling} provides additional details for the implementation of $\PT$ both for the geometrically local and non-local cases. 
\Cref{app:error-metric} details how the diamond-norm distance used for the numerical studies is obtained.
In particular, we explain how it relates to a phase-optimised operator norm, which it is used for the results reported.

\section{Dynamics of Trotter error}
\label{app:trotter_error}

In this appendix, we derive the nested-commutator representation of the remainder Hamiltonian $G_k$ that generates the Trotter error, which can be corrected by applying the correction channel $Q_k$.
\Cref{app:trotter-recap} introduces the notations and the Trotter formulas used throughout.
\Cref{app:Gk-definition} defines $G_k$ through the factorisation $U(\tau)=Q_k(\tau)\,\mathcal{S}_k(\tau)$ and provides a first expression for the remainder Hamiltonian.
\Cref{app:conjugation-identity} introduces the conjugation identity, central to the following derivations. The remaining subsections apply this tool to obtain $G_k$ at successive orders. As working examples, and focusing on the case $H=A+B$, \cref{app:G1-kernel} derives the expression for $G_1$ that was presented in the main text, Eq.~(1), while \cref{app:k2-corollary} obtains $G_2$. Finally, \cref{app:higher-order} proves the general kernel-form statement for $G_k$ of Lemma~2.

\subsection{Suzuki--Trotter formulas}
\label{app:trotter-recap}

We recall the Suzuki--Trotter product formulas~\cite{Suzuki1991, Childs2021trotter} that are used later on.
For a Hamiltonian decomposed as 
\begin{equation}
    H=\sum_{\ell=1}^{L}H_\ell,
\end{equation}
with each $H_\ell$ being a weighted sum of mutually commuting Pauli strings, the first- and second-order formulas read
\begin{equation}
\begin{aligned}
  \mathcal{S}_1(\tau) &:= \prod_{\ell=1}^{L}\e^{-\im\tau H_\ell}, \\
  \mathcal{S}_2(\tau) &:= \prod_{\ell=1}^{L}\e^{-\im\frac{\tau}{2} H_\ell}\,\prod_{\ell=L}^{1}\e^{-\im\frac{\tau}{2} H_\ell}.
\end{aligned}
  \label{eq:S1-S2-recap}
\end{equation}
Higher even orders are obtained through the Suzuki recursion
\begin{equation}
\begin{aligned}
  \mathcal{S}_{2k+2}(\tau)
  &:= \mathcal{S}_{2k}(u_k\tau)^2\,\mathcal{S}_{2k}\bigl((1-4u_k)\tau\bigr)\,\mathcal{S}_{2k}(u_k\tau)^2, \\
  u_k &:= \bigl(4-4^{1/(2k+1)}\bigr)^{-1}.
\end{aligned}
  \label{eq:Suzuki-recursion}
\end{equation}
Let us refer to any unitary of the form $\exp(-\im\alpha \widetilde H_j)$, for an arbitrary angle $\alpha$, as \emph{a layer}. Note that such a layer can always be implemented as a product of single-Pauli rotations, since $H_{\ell_j}$ is a weighted sum of mutually commuting Pauli strings.
We further define the number of \emph{stages}~\cite{Childs2021trotter} as $\Upsilon_k = 2\cdot 5^{k/2-1}$ (for even $k\geq 2$) and $\Upsilon_1=1$, such that each stage involves a product of $L$ \emph{layers}. Hence, a $k$-th order formula $\mathcal{S}_k$ has $\tilde\Upsilon_k:=\Upsilon_k L$ layers and, for a Hamiltonian $H$ with $N$ Pauli terms, is implemented in terms of $\Upsilon_k N$ Pauli rotations.

For convenience, we recast a $k$-th order formula as
\begin{equation}
\begin{aligned}
  \mathcal{S}_k(\tau) &= \prod_{j=1}^{\tilde\Upsilon_k}\e^{-\im\tau \widetilde H_j} \quad (\textrm{with}\, \sum_{j=1}^{\tilde\Upsilon_k}\widetilde H_j= H),
\end{aligned}
  \label{eq:Sk-product-recap}
\end{equation}
where each $\widetilde H_j = \omega_j H_{\ell_j}$ is a weighted multiple of one of the components of the Hamiltonian, with Suzuki coefficients $\{\omega_j\}$ and indices $\{\ell_j\}\subseteq\{1,\dots,L\}$ fixed by the product formula. 

Let us also fix the locality notions used later, in particular in Sec.~C. The $n$ qubits sit on the vertices of a lattice of fixed spatial dimension and bounded degree, so each site has $\mathcal{O}(1)$ neighbours. We call $H$ \emph{geometrically local} when each of its Pauli terms is supported within a ball of constant radius, independent of the system size, which we call the \emph{locality range}. We further assume that the Pauli coefficients of $H$ are bounded in magnitude by a constant. We collectively refer to the spatial dimension, the lattice degree, the locality range, and the coupling-strength bound as the \emph{locality parameters}. For a geometrically local $H$, the number of Pauli terms satisfies $N=\mathcal{O}(n)$.

Finally, we recall the required resources to simulate the Hamiltonian $H$, up to some error $\epsilon$, with  Suzuki--Trotter formulas.
The Hamiltonian $H$ can be further decomposed as a sum of $N$ Pauli strings $H=\sum_{j=1}^Nc_jP_j$ with corresponding 1-norm $\beta:=\sum_{j=1}^N|c_j|$. Then, the algorithmic bias of the $k$-th order Trotter formula scales as
\begin{align}
\label{eq:kth_order_trotter_error}
    \big\|\mathcal{S}_k(\tau) - U(\tau)\big\| 
    = 
    \mathcal{O}\big((\beta \tau)^{k+1}\big).
\end{align}
This follows from Lemma~6 in Ref.~\cite{Childs2021trotter} together with the inequality $\sum_{\ell=1}^L\|H_\ell\|\le \beta$ (we assume that the same Pauli strings do not appear in multiple $H_\ell$'s so that the inequality holds).
Note that the error~\eqref{eq:kth_order_trotter_error} indicates that the $k$-th order correction unitary satisfies $\mathcal{S}_k(\tau)\,U^\dagger(\tau)=\mathbb{I}+\mathcal{O}(\tau^{k+1})$ at small $\tau$. %
It follows that the Hamiltonian $H$ can be simulated for a time $t$ and an error $\epsilon$ by splitting the time $t$ into a number of steps (Corollary~7 in \cite{Childs2021trotter}):
\begin{equation}
r=\mathcal{O}\left(\frac{(\beta t)^{1+1/k}}{\epsilon^{1/k}}\right),
\end{equation}
leading to the total number of rotation gates quoted in the main text:
\begin{align}
    \mathcal{O}(rN) 
    = 
    \mathcal{O}\left(\frac{N(\beta t)^{1+1/k}}{\epsilon^{1/k}}\right).
\end{align}

In many cases, the required resources of the Trotter formula can be bounded more tightly by leveraging the Hamiltonian structure through the nested commutator $\alpha_{\rm comm}:=\sum_{\ell_1,\ell_2,\dots,\ell_{k+1}=1}^{L}\big\|[H_{\ell_{k+1}},\dots[H_{\ell_2},H_{\ell_1}]]\big\|$ (Corollary~12 in \cite{Childs2021trotter}):
\begin{align}
    r=\mathcal{O}\left(\frac{\alpha_{\rm comm}^{1/k}t^{1+1/k}}{\epsilon^{1/k}}\right).
\end{align}
In particular, $\alpha_{\rm comm}=\mathcal{O}(n)$ and $N=\mathcal{O}(n)$ when the Hamiltonian is geometrically local, resulting in the upper bound of the rotation gate count,
\begin{align}
    \mathcal{O}(rN) 
    = \mathcal{O}\left(\frac{(n t)^{1+1/k}}{\epsilon^{1/k}}\right),
\end{align}
quoted in the main text.

\subsection{Definition of the remainder Hamiltonian \texorpdfstring{$G_k$}{Gk}}
\label{app:Gk-definition}

In this section we derive a first expression for the remainder Hamiltonian $G_k$ at arbitrary $k$ and discuss its properties.
Splitting the total evolution time $t$ into $r$ steps of length $\tau=t/r$, we define within a single step the \emph{correction unitary} as a function of the running time $s\in[0,\tau]$,
\begin{equation}
    Q_k(s) := U(s)\,\mathcal{S}_k^\dagger(s),
    \label{eq:Qk-def-app}
\end{equation}
so that $U(s)=Q_k(s)\,\mathcal{S}_k(s)$ and $Q_k(0)=\mathbb{I}$. Setting $s=\tau$ recovers the per-step correction.
The adjoint
\begin{equation}
    Q_k^\dagger(s) = \mathcal{S}_k(s)\,U^\dagger(s)
    \label{eq:Qkdag-def-app}
\end{equation}
captures the deviation between the product formula and the exact evolution as a single unitary, which we call the \emph{Trotter error}.
Differentiating \cref{eq:Qkdag-def-app} and using $\im\partial_s U^\dagger=-U^\dagger H$ yields
\begin{equation}
\begin{split}
    \im\partial_s Q_k^\dagger(s)
    &= \im\bigl(\partial_s\mathcal{S}_k(s)\bigr)U^\dagger(s)
       + \mathcal{S}_k(s)\,\im\partial_s U^\dagger(s) \\
    &= H_S(s)\,Q_k^\dagger(s) - Q_k^\dagger(s)\,H,
\end{split}
\label{eq:Qkdag-derivative}
\end{equation}
where we have defined
\begin{equation}\label{eq:time_remainder_hamiltonian}
H_S(s):=\im(\partial_s\mathcal{S}_k)\mathcal{S}_k^\dagger(s),
\end{equation}
which is the time-dependent generator of the Trotter unitary $\mathcal{S}_k(s)$.
Since $U^\dagger(s) H U(s) = H$ for any $s$, we have $Q_k^\dagger(s)\,H\,Q_k(s) = \mathcal{S}_k(s)\,H\,\mathcal{S}_k^\dagger(s)$. Thus, we obtain
\begin{equation}
    \im\partial_s Q_k^\dagger(s) = G_k(s)\,Q_k^\dagger(s),
    \label{eq:Qkdag-left-mul}
\end{equation}
with the \emph{remainder Hamiltonian} defined as
\begin{equation}
    G_k(s) := H_S(s) - \mathcal{S}_k(s)\,H\,\mathcal{S}_k^\dagger(s).
    \label{eq:Gk-left-mul-def}
\end{equation}
The remainder Hamiltonian inherits the accuracy of the Trotter formula it is built from.
\begin{lemma}%
\label{lem:Gk-order}
For any $k$-th order Trotter formula $\mathcal{S}_k$, the remainder Hamiltonian satisfies $G_k(s) = \mathcal{O}(s^k)$.
\end{lemma}
\begin{proof}
The defining property of an order-$k$ Trotter formula~\cite{Suzuki1991, Childs2021trotter} is the small-$s$ identity
\begin{equation}
    \mathcal{S}_k(s)\,\e^{\im Hs} = \mathbb{I} + \mathcal{O}(s^{k+1}).
    \label{eq:Sk-order-property}
\end{equation}
Differentiating in $s$, multiplying by $\im$, and right-multiplying by $\e^{-\im Hs}\mathcal{S}_k^\dagger(s)$ gives $\im(\partial_s\mathcal{S}_k)\mathcal{S}_k^\dagger - \mathcal{S}_k H\mathcal{S}_k^\dagger = \mathcal{O}(s^k)$, which is $G_k(s)$ as per~\cref{eq:Gk-left-mul-def}.
\end{proof}
Since $Q_k^\dagger(0)=\mathbb{I}$, \cref{eq:Qkdag-left-mul} is solved through the time-ordered exponential $Q_k^\dagger(s)=\mathcal{T}\exp\bigl(-\im\int_0^{s}G_k(s')\,\diff s'\bigr)$.

\subsection{Tool: the conjugation identity}
\label{app:conjugation-identity}

The following derivations of $G_k$ in \cref{eq:Gk-left-mul-def} at first ($k=1$), second ($k=2$), and general order $k$ all rely on a single
operator-calculus identity: the integral-remainder expansion of
$\e^{-\im t\,\mathrm{ad}_X}(Y)$, which we now detail.
For an arbitrary operator $X$, the commutator map $\mathrm{ad}_X$ is defined by
\begin{equation}
  \mathrm{ad}_X(Y) := [X, Y] = XY - YX,
  \label{eq:ad-def}
\end{equation}
and the conjugation map by
\begin{equation}
  \mathrm{Ad}_U(Y) := U\,Y\,U^\dagger.
  \label{eq:Ad-def}
\end{equation}
The $n$-fold commutator is denoted by $\mathrm{ad}_X^n$ such that $\mathrm{ad}_X^2(Y) = [X,[X,Y]]$. Furthermore, recall that conjugation by an exponential can be expressed through exponentiation of the commutator: $\mathrm{Ad}_{\e^X}(Y) = \e^{X}Y\e^{-X} = \e^{\mathrm{ad}_X}(Y)$.
Applying a Taylor expansion (with integral remainder) at order $p$  to the operator-valued function $f(t) := \e^{-\im t\,\mathrm{ad}_X}(Y)$, one obtains
\begin{equation}
  \e^{-\im t\,\mathrm{ad}_X}(Y)
  =
  \sum_{j=0}^{p} \frac{(-\im t)^j}{j!}\mathrm{ad}_X^j(Y)
  + (-\im)^{p+1}\int_0^t \frac{(t-\sigma)^p}{p!}\,\e^{-\im\sigma\,\mathrm{ad}_X}\bigl(\mathrm{ad}_X^{p+1}(Y)\bigr)\,\diff\sigma.
  \label{eq:Ad-expand}
\end{equation}
For the first three truncation orders that we will use, we get
\begin{align}
  p=0:\quad
  \e^{-\im t\,\mathrm{ad}_X}(Y) &= Y - \im\!\int_0^t \e^{-\im\sigma\,\mathrm{ad}_X}\bigl([X,Y]\bigr)\,\diff\sigma,
  \label{eq:Ad-expand-p0} \\
  p=1:\quad
  \e^{-\im t\,\mathrm{ad}_X}(Y) &= Y - \im t\,[X,Y] - \int_0^t (t-\sigma)\,\e^{-\im\sigma\,\mathrm{ad}_X}\bigl([X,[X,Y]]\bigr)\,\diff\sigma,
  \label{eq:Ad-expand-p1} \\
  p=2:\quad
  \e^{-\im t\,\mathrm{ad}_X}(Y) &= Y - \im t\,[X,Y] - \frac{t^2}{2}\,[X,[X,Y]]
  + \im\!\int_0^t \frac{(t-\sigma)^2}{2}\,\e^{-\im\sigma\,\mathrm{ad}_X}\bigl([X,[X,[X,Y]]]\bigr)\,\diff\sigma.
  \label{eq:Ad-expand-p2}
\end{align}

\subsection{First order: \texorpdfstring{$G_1$}{G1}}
\label{app:G1-kernel}

At first order, $\mathcal{S}_1(s)=\e^{-\im As}\e^{-\im Bs}$. The Trotter generator is given by
\begin{equation}
  H_S(s)
  = \im(\partial_s\mathcal{S}_1)\mathcal{S}_1^\dagger
  =  \e^{-\im s\, \mathrm{ad}_A}(A + B).
\end{equation}
Likewise, the conjugated Hamiltonian decomposes as
\begin{equation}
  \mathcal{S}_1(s)\,H\,\mathcal{S}_1^\dagger(s)
  = \e^{-\im s\,\mathrm{ad}_A}\!\bigl(\e^{-\im s\,\mathrm{ad}_B}(A) + B\bigr).
\end{equation}
Substituting the two previous equalities into \cref{eq:Gk-left-mul-def}, one gets
\begin{equation}
  G_1(s)
  = -\,\e^{-\im s\,\mathrm{ad}_A}\!\bigl(\e^{-\im s\,\mathrm{ad}_B}(A) - A\bigr).
\end{equation}
Applying \cref{eq:Ad-expand-p0} with $X=B$, $Y=A$ gives 
\begin{equation}
\e^{-\im s\,\mathrm{ad}_B}(A) - A = -\im\!\int_0^s \e^{-\im\sigma\,\mathrm{ad}_B}([B,A])\,\diff\sigma,    
\end{equation}
and we finally obtain
\begin{equation}
  G_1(s)
  = -\int_0^s \e^{-\im As}\,\e^{-\im B\sigma}\,\im[A,B]\,
      \e^{\im B\sigma}\,\e^{\im As}\,\diff\sigma,
  \label{eq:G1-kernel-app}
\end{equation}
which is~Eq.~(1) of the main text.

\subsection{Second order: \texorpdfstring{$G_2$}{G2}}
\label{app:k2-corollary}

We now derive $G_2(s)$ for the second-order Trotter formula $\mathcal{S}_2(s)=\e^{-\im A\frac{s}{2}}\,\e^{-\im Bs}\,\e^{-\im A\frac{s}{2}}$. %
Differentiating $\mathcal{S}_2(s)$ gives
\begin{equation}
    H_S(s)= \im(\partial_s\mathcal{S}_2)\mathcal{S}_2^\dagger=\frac{1}{2}A+\e^{-\im \frac{s}{2}\mathrm{ad}_A}(B)+\frac{1}{2}\e^{-\im \frac{s}{2}\,\mathrm{ad}_A}\!\bigl(\e^{-\im s\,\mathrm{ad}_B}(A)\bigr).
\end{equation}
In addition, the conjugated Hamiltonian for $H=A+B$ decomposes as
\begin{equation}
  \mathcal{S}_2(s)\,H\,\mathcal{S}_2^\dagger(s)
  = \e^{-\im \frac{s}{2}\,\mathrm{ad}_A}\e^{-\im s\,\mathrm{ad}_B}\!\bigl( A + \e^{-\im \frac{s}{2}\,\mathrm{ad}_A}(B)\bigr).
\end{equation}
In turn, \cref{eq:Gk-left-mul-def} yields
\begin{equation}
\begin{split}
  G_2(s)
  &= \frac{1}{2}\,\e^{-\im \frac{s}{2}\,\mathrm{ad}_A}\!\bigl(A-\e^{-\im s\,\mathrm{ad}_B}(A)\bigr) \\
  &\quad
   + \e^{-\im \frac{s}{2}\,\mathrm{ad}_A}\!\bigl(B-\e^{-\im s\,\mathrm{ad}_B}(\e^{-\im \frac{s}{2}\,\mathrm{ad}_A}(B))\bigr).
\end{split}
  \label{eq:G2-Ad-grouped}
\end{equation}
Applying \cref{eq:Ad-expand-p1} to the expression in each parenthesis, we obtain
\begin{align}
  A-\e^{-\im s\,\mathrm{ad}_B}(A)
  &= -\im s[A,B]
    - \int_0^s (s-\sigma)\,\e^{-\im\sigma\,\mathrm{ad}_B}\bigl([B,[A,B]]\bigr)\,\diff\sigma,
  \label{eq:expand-A-term} \\
  B-\e^{-\im s\,\mathrm{ad}_B}\!\bigl(\e^{-\im s/2\,\mathrm{ad}_A}(B)\bigr)
  &= \e^{-\im s\,\mathrm{ad}_B}\!\biggl(
       \im\frac{s}{2}[A,B]
     + \int_0^{s/2}\!\!\Bigl(\frac{s}{2}-\sigma\Bigr)
        \e^{-\im\sigma\,\mathrm{ad}_A}\bigl([A,[A,B]]\bigr)\,\diff\sigma\biggr).
  \label{eq:expand-B-term}
\end{align}
Substituting back into \cref{eq:G2-Ad-grouped} and applying \cref{eq:Ad-expand-p0} to $\e^{-\im s\,\mathrm{ad}_B}([A,B])$ we see that the two $(\im s/2)[A,B]$ contributions cancel. Recombining the surviving integrals, the half-step weights combine as $\tfrac{s}{2}-\tfrac12(s-\sigma)=\tfrac{\sigma}{2}$, which is the scalar weight on the $\Xi_B$ integral, while the $\hat\Xi_A$ integral keeps the weight $\tfrac{s}{2}-\sigma$.
We finally obtain
\begin{equation}
  G_2(s)
  = \int_0^{s/2}\!\!\Bigl(\frac{s}{2}-\sigma\Bigr)\,
      \hat\Xi_A(s,\sigma)\,\diff\sigma
    + \frac{1}{2}\int_0^{s} \sigma\,\Xi_B(s,\sigma)\,\diff\sigma,
  \label{eq:G2-exact}
\end{equation}
where the rotated nested commutators are defined as
\begin{align}
  \hat\Xi_A(s,\sigma)
    &:= \e^{-\im s/2\,\mathrm{ad}_A}\!\bigl(
          \e^{-\im s\,\mathrm{ad}_B}(
            \e^{-\im\sigma\,\mathrm{ad}_A}([A,[A,B]]))\bigr),
    \label{eq:Xi-A-hat} \\
  \Xi_B(s,\sigma)
    &:= \e^{-\im s/2\,\mathrm{ad}_A}\!\bigl(
          \e^{-\im\sigma\,\mathrm{ad}_B}([B,[A,B]])\bigr).
    \label{eq:Xi-B}
\end{align}
Both scalar weights in \cref{eq:G2-exact} are $\mathcal{O}(s)$ over an integration domain of width $\mathcal{O}(s)$, making the $\mathcal{O}(s^2)$ bound on $G_2$ explicit and consistent with \cref{thm:Gk-commutator}.

\subsection{General \texorpdfstring{$k$}{k}: commutator kernel form}
\label{app:higher-order}

This subsection proves Lemma~2 of the main text, and proceeds in three parts. First, we rewrite $G_k(s)$ as a sum of conjugated Hamiltonian terms. Second, we recall the commutator-expansion theorem of Childs et al.~\cite{Childs2021trotter}, which expands a conjugation by a product of matrix exponentials as a polynomial in $s$ of degree $k{-}1$ together with a remainder consisting of conjugated $k$-fold nested commutators with norm scaling as $\mathcal{O}(s^k)$. Third, we apply that expansion to each of the conjugations: their polynomial parts combine into a polynomial in $s$ of degree $k{-}1$ which must vanish since $G_k(s)=\mathcal{O}(s^k)$, leaving only terms of order $s^k$ and higher. These terms are then collected, yielding the kernel form (Lemma~2) stated in the main text.

Recall from~\cref{eq:Sk-product-recap} that the $k$-th order Trotter formula is expressed as $\mathcal{S}_k(s) = \prod_{j=1}^{\tilde\Upsilon_k} \e^{-\im s \widetilde H_j}$, with each $\widetilde H_j = \omega_j H_{\ell_j}$ being a weighted component from the Hamiltonian splitting $H=\sum_\ell H_\ell$. Define the forward partial products as
\begin{equation}
  L_k^j(s) := \prod_{j'=1}^{j-1} \e^{-\im s \widetilde H_{j'}},
  \label{eq:Lk-def}
\end{equation}
with $L_k^1(s):=\mathbb{I}$.
Recall from~\cref{app:Gk-definition} that the remainder Hamiltonian can be expressed as
$G_k = H_S - \mathcal{S}_kH\mathcal{S}_k^\dagger$. The first term
can be recast as
\begin{equation}
  H_S(s)
  = \im(\partial_s\mathcal{S}_k)\mathcal{S}_k^\dagger
  = \sum_{j=1}^{\tilde\Upsilon_k} \mathrm{Ad}_{L_k^j(s)}(\widetilde H_j).
\end{equation}
Since $\sum_j \widetilde H_j=H$ for a consistent product formula~\eqref{eq:Sk-product-recap}, we further have $\mathcal{S}_k H\mathcal{S}_k^\dagger = \sum_j\mathrm{Ad}_{\mathcal{S}_k(s)}(\widetilde H_j)$. Combining the two expressions gives
\begin{equation}
  G_k(s)
  = \sum_{j=1}^{\tilde\Upsilon_k}\bigl[\mathrm{Ad}_{L_k^j(s)}(\widetilde H_j)
  - \mathrm{Ad}_{\mathcal{S}_k(s)}(\widetilde H_j)\bigr].
  \label{eq:Gk-diff-form}
\end{equation}

To proceed further, we expand it as a polynomial in $s$ and rotated nested commutators.
The tool for this is Theorem~10 of Childs et al.~\cite{Childs2021trotter}, restated below in our conventions.

\begin{lemma}[Commutator expansion of conjugated exponentials; \cite{Childs2021trotter}, Thm.~10, in our convention]
\label{lem:conj-expansion}
For Hermitian operators $X_1,\ldots,X_n$, an arbitrary operator $Y$, and a positive integer $p$,
\begin{equation}
  \e^{-\im s\, \mathrm{ad}_{X_1}}\cdots \e^{-\im s\, \mathrm{ad}_{X_n}}(Y)
  = \sum_{r=0}^{p-1} D_r\,s^r + \mathcal{N}_p(s),
  \label{eq:conj-expansion}
\end{equation}
where the operators $D_0,\ldots,D_{p-1}$ are independent of $s$ and the remainder
\begin{equation}
\begin{split}
  \mathcal{N}_p(s)
  &= (-\im)^p \sum_{i=1}^{n}\;\sum_{\substack{q_i+q_{i+1}+\cdots+q_n = p \\ q_i \ge 1}}
  \int_0^{s}\!\diff\sigma\;
  \e^{-\im s\,\mathrm{ad}_{X_1}}\cdots \e^{-\im s\,\mathrm{ad}_{X_{i-1}}}\,\e^{-\im\sigma\,\mathrm{ad}_{X_i}}
  \!\left(
  \mathrm{ad}_{X_i}^{q_i}\,\mathrm{ad}_{X_{i+1}}^{q_{i+1}}\cdots\mathrm{ad}_{X_n}^{q_n}(Y)
  \right) \\
  &\qquad\qquad\times
  \frac{(s-\sigma)^{q_i-1}\,s^{q_{i+1}+\cdots+q_n}}
       {(q_i{-}1)!\,q_{i+1}!\cdots q_n!}.
\end{split}
  \label{eq:Np-explicit}
\end{equation}
satisfies $\|\mathcal{N}_p(s)\|=\mathcal{O}(s^p)$ as $s\to 0$.
\end{lemma}

The first term on the right of \cref{eq:conj-expansion} is a polynomial in $s$ of degree $p-1$, while the remainder $\mathcal{N}_p(s)$ is $\mathcal{O}(s^p)$. Each summand of $\mathcal{N}_p$ carries a $p$-fold nested commutator $\mathrm{ad}_{X_i}^{q_i}\,\mathrm{ad}_{X_{i+1}}^{q_{i+1}}\cdots\mathrm{ad}_{X_n}^{q_n}(Y)$, conjugated by $i$ unitaries.
The relation \cref{eq:conj-expansion} decomposes a conjugated operator into terms scaling as $\mathcal{O}(s^r)$ for $r<p$, and remaining $\mathcal{O}(s^p)$.
As we will see, these polynomial parts drop out of $G_k$, leaving only the remainders,  sums of conjugated $k$-fold nested commutators.
To proceed further, we tweak the product of \cref{eq:Lk-def} to define
\begin{equation}
  L_k^j(s,\sigma)
  := \Bigl(\prod_{j'=1}^{j-1}\e^{-\im s \widetilde H_{j'}}\Bigr)\,\e^{-\im\sigma \widetilde H_j},
  \label{eq:Lk-two-arg}
\end{equation}
for $1\le j\le \tilde\Upsilon_k$, so that $L_k^j(s,0) = L_k^j(s)$ and $L_k^j(s,s) = L_k^{j+1}(s)$. The main-text Lemma~2 writes this partial product as $L_\nu$.

We now apply \cref{lem:conj-expansion} for $p=k$ to the conjugations in~\cref{eq:Gk-diff-form}.
On one hand, the terms in $\mathrm{Ad}_{L_k^j(s)}(\widetilde H_j)$ (we call them the $L$-side) are conjugations by $j-1$ exponentials generated by operators $\widetilde H_1,\dots,\widetilde H_{j-1}$, so \cref{eq:conj-expansion,eq:Np-explicit} applies with $n=j-1$ and $Y=\widetilde H_j$.
On the other hand, the terms in $\mathrm{Ad}_{\mathcal{S}_k(s)}(\widetilde H_j)$ (we call them the $S$-side) are conjugations by $\tilde\Upsilon_k$ exponentials generated by operators $\widetilde H_1,\dots,\widetilde H_{\tilde\Upsilon_k}$, so \cref{eq:conj-expansion,eq:Np-explicit} applies with $n=\tilde\Upsilon_k$ and $Y=\widetilde H_j$.
Thus these terms are expressed as
\begin{equation}
  \mathrm{Ad}_{L_k^j(s)}(\widetilde H_j) = \mathcal{P}_j^{L}(s) + R_j^{L}(s),
  \qquad
  \mathrm{Ad}_{\mathcal{S}_k(s)}(\widetilde H_j) = \mathcal{P}_j^{S}(s) + R_j^{S}(s),
  \label{eq:two-side-poly-rem}
\end{equation}
where $\mathcal{P}_j^{L}(s),\mathcal{P}_j^{S}(s)$ are operator-valued polynomials of degree at most $k-1$ in $s$, and the integral remainders take the form
\begin{equation}
\begin{split}
  R_j^{L}(s)
  &= (-\im)^k \,\!\!\sum_{i=1}^{j-1}\sum_{\substack{q_i+q_{i+1}+\cdots+q_{j-1}=k\\ q_i\ge1}}
  \int_0^s \!\diff\sigma\,
  \frac{(s-\sigma)^{q_i-1}\,s^{q_{i+1}+\cdots+q_{j-1}}}
       {(q_i{-}1)!\,q_{i+1}!\cdots q_{j-1}!}\,
  \mathrm{Ad}_{L_k^{i}(s,\sigma)}\!\bigl(\mathrm{ad}_{\widetilde H_i}^{q_i}\,\mathrm{ad}_{\widetilde H_{i+1}}^{q_{i+1}}\cdots\mathrm{ad}_{\widetilde H_{j-1}}^{q_{j-1}}(\widetilde H_j)\bigr),
  \\
  R_j^{S}(s)
  &= (-\im)^k \,\!\!\sum_{i=1}^{\tilde\Upsilon_k}\sum_{\substack{q_i+q_{i+1}+\cdots+q_{\tilde\Upsilon_k}=k\\ q_i\ge1}}
  \int_0^s \!\diff\sigma\,
  \frac{(s-\sigma)^{q_i-1}\,s^{q_{i+1}+\cdots+q_{\tilde\Upsilon_k}}}
       {(q_i{-}1)!\,q_{i+1}!\cdots q_{\tilde\Upsilon_k}!}\,
  \mathrm{Ad}_{L_k^{i}(s,\sigma)}\!\bigl(\mathrm{ad}_{\widetilde H_i}^{q_i}\,\mathrm{ad}_{\widetilde H_{i+1}}^{q_{i+1}}\cdots\mathrm{ad}_{\widetilde H_{\tilde\Upsilon_k}}^{q_{\tilde\Upsilon_k}}(\widetilde H_j)\bigr).
\end{split}
  \label{eq:R-side-explicit}
\end{equation}
Each integrand is the unitary conjugation of a $k$-fold nested commutator of the components $H_\ell$, hence has spectral norm bounded by that commutator's norm, and the scalar weight integrates to a prefactor of order $s^{q_i+q_{i+1}+\cdots}=s^k$, so $\|R_j^{L}(s)\|,\,\|R_j^{S}(s)\|=\mathcal{O}(s^k)$ uniformly in $j$.

\begin{theorem}[Exact nested-commutator expression for $G_k$]
\label{thm:Gk-commutator}
Let $\mathcal{S}_k$ be any $k$-th order Trotter formula.
Then the remainder Hamiltonian is given by
\begin{equation}
  G_k(s) = \sum_{j=1}^{\tilde\Upsilon_k}\bigl(R_j^{L}(s)-R_j^{S}(s)\bigr),
  \label{eq:Gk-Nkl}
\end{equation}
with $R_j^{L}$ and $R_j^{S}$ defined in~\cref{eq:R-side-explicit}.
\end{theorem}
\begin{proof}
Substituting the polynomial-plus-remainder expansions~\cref{eq:two-side-poly-rem} into \cref{eq:Gk-diff-form} gives
\begin{equation}
  G_k(s)
  = \mathcal{P}_{<k}(s)
  + \sum_{j=1}^{\tilde\Upsilon_k}\bigl[R_j^{L}(s) - R_j^{S}(s)\bigr],
  \label{eq:Gk-poly-plus-remainder}
\end{equation}
where $\mathcal{P}_{<k}(s)$ is a polynomial in $s$ of degree at most $k{-}1$.
The remainder sum is $\mathcal{O}(s^k)$ by the explicit form \cref{eq:R-side-explicit}: each conjugation has spectral norm one, and $\int_0^s K_\nu(s,\sigma)\,\diff\sigma$ contributes a prefactor of order $s^k$.
By \cref{lem:Gk-order} the left-hand side is $\mathcal{O}(s^k)$, so $\mathcal{P}_{<k}(s)$ is itself $\mathcal{O}(s^k)$.
A polynomial of degree at most $k{-}1$ that is $\mathcal{O}(s^k)$ as $s\to 0$ is identically zero, hence $\mathcal{P}_{<k}\equiv0$ and only the remainder sum survives, yielding \cref{eq:Gk-Nkl}.
\end{proof}

Inserting~\cref{eq:R-side-explicit} into~\cref{eq:Gk-Nkl} and collecting every summand under a single multi-index $\nu$ gives the direct kernel form
\begin{align}
  G_k(s) = \sum_\nu\int_0^s
    K_\nu(s,\sigma)\,
    \mathrm{Ad}_{L_k^{j_\nu}(s,\sigma)}(C_\nu)
    \,\diff\sigma,
\label{eq:Gk-remainder-explicit}
\end{align}
where $\nu$ enumerates all nested commutators produced by the two conjugation expansions in \cref{eq:R-side-explicit}: it is the multi-label
\begin{equation}
  \nu = (T,\,j,\,i,\,\vec q),
  \qquad
  T\in\{L,S\},
  \quad
  1\le j \le \tilde\Upsilon_k,
  \quad
  1\le i \le n_T,
  \quad
  \vec q=(q_i,q_{i+1},\dots,q_{n_T}),
  \label{eq:nu-label}
\end{equation}
with components non-negative integers, $q_i\ge1$, and $q_i+q_{i+1}+\cdots+q_{n_T}=k$. Here $T$ selects the $L$-side ($n_L=j-1$) or $S$-side ($n_S=\tilde\Upsilon_k$), $j$ is the parent rotation index from \cref{eq:Gk-diff-form} (the factor $\widetilde H_j$ being conjugated), and $i$ is the inner sum index of \cref{eq:R-side-explicit} (the layer whose angle is the integration variable). Adopting the notation of Lemma~2 in the main text, we write
\begin{equation}
  j_\nu := i,
  \label{eq:jnu-projection}
\end{equation}
such that the $L_k^{i}(s,\sigma)$ terms in \cref{eq:R-side-explicit} become $L_k^{j_\nu}(s,\sigma)$ in \cref{eq:Gk-remainder-explicit}.
In addition, we defined
\begin{equation}
  K_\nu(s,\sigma)
  =
  \frac{(s-\sigma)^{q_{j_\nu}-1}\,s^{q_{j_\nu+1}+\cdots+q_{n_T}}}
       {(q_{j_\nu}{-}1)!\,q_{j_\nu+1}!\cdots q_{n_T}!},
  \label{eq:Knu-explicit}
\end{equation}
so $K_\nu(s,\sigma)\ge0$ for $0\le\sigma\le s$ and is a homogeneous polynomial of degree $k{-}1$ in $(s,\sigma)$.
Finally, the operator $C_\nu$ collects the $(-\im)^k$ phase of \cref{eq:R-side-explicit}, the sign $\varepsilon_L=+1$, $\varepsilon_S=-1$ from the $R_j^L-R_j^S$ split, and the Suzuki coefficients pulled out of $\mathrm{ad}_{\widetilde H_p}^{q_p}=\omega_p^{q_p}\,\mathrm{ad}_{H_{\ell_p}}^{q_p}$ and out of the innermost argument $\widetilde H_j=\omega_j H_{\ell_j}$:
\begin{equation}
  C_\nu
  = \varepsilon_T\,(-\im)^k\,\omega_j\,
    \Bigl(\prod_{p=j_\nu}^{n_T} \omega_p^{q_p}\Bigr)\,
    \mathrm{ad}_{H_{\ell_{j_\nu}}}^{q_{j_\nu}}\,
    \mathrm{ad}_{H_{\ell_{j_\nu+1}}}^{q_{j_\nu+1}}\cdots
    \mathrm{ad}_{H_{\ell_{n_T}}}^{q_{n_T}}(H_{\ell_j}).
  \label{eq:Cnu-explicit}
\end{equation}
This is the kernel form referred to in Lemma~2.
The kernel $K_\nu$ is non-negative and a homogeneous polynomial of degree $k-1$ in $(s,\sigma)$, hence
\begin{equation}
    \int_0^s K_\nu(s,\sigma)\,\diff\sigma = s^k w_\nu,
    \label{eq:wnu-int}
\end{equation}
where
\begin{equation}
  w_\nu := \int_0^1 K_\nu(1,\sigma)\,\diff\sigma
  = \frac{1}{q_{j_\nu}!\,q_{j_\nu+1}!\cdots q_{n_T}!}\,\ge 0.
  \label{eq:wnu-def}
\end{equation}
Together with \cref{thm:Gk-commutator}, this establishes Lemma~2, including the bound $\|G_k(s)\|=\mathcal{O}(s^k)$.

\section{Quasi-Probabilistic Decomposition for Hamiltonian Simulation}
\label{app:tepai-details}

In this section, we provide additional details regarding TE-PAI~\cite{kiumi2025te} used to implement the correction channels, whose generators were derived in the previous appendix.
First, in~\Cref{app:tepai-task} we review the task that TE-PAI accomplishes, namely unbiased Hamiltonian simulation (as a channel) of a time-dependent Hamiltonian.
Along the way, we detail concepts of quasi-probabilistic decomposition, sampling overhead, and Hermitian involution decomposition. Next we review the working details of TE-PAI in \Cref{app:tepai-algo}, for which a detailed implementation is provided in \cref{app:tepai-impl}. Finally in~\Cref{app:tepai-resources} we report the resource counts of TE-PAI supporting~Lemma~1 of the main text.

\subsection{Unbiased channel simulation and Hermitian-involution decompositions}
\label{app:tepai-task}
Consider a time-dependent Hamiltonian $H(s)$ and the corresponding Hamiltonian evolution (as a unitary) for a time $t$, 
\begin{equation}
    U(t) = \mathcal{T}\exp\!\left(-\im\int_0^{t}H(s)\,\diff s\right).
\end{equation}
Further define the evolution channel $\mathcal{U}(t)[\cdot] = U(t)(\cdot)U(t)^{\dag}$.
For now, let us drop the time dependence for ease of notation.
We say that a set $\{\gamma_l, \mathcal{V}_l\}$, of coefficients $\gamma_l\in \mathbb{R}$ and quantum (completely positive and trace preserving) channels $\mathcal{V}_l$, is a \emph{quasi-probabilistic decomposition} (QPD) of $\mathcal{U}$ whenever it satisfies
\begin{equation}
  \mathcal{U} = \sum_l \gamma_l \,\mathcal{V}_l.
  \label{eq:qpd-abstract}
\end{equation}

For tasks of estimating expectation values of an evolved state, QPDs can be used. To see this, first recast the decomposition in~\cref{eq:qpd-abstract}, as
\begin{equation}
  \mathcal{U} = \gamma \sum_l p_l\,\mathrm{sgn}(\gamma_l)\,\mathcal{V}_l = \mathbb{E}_l[\xi_l\,\mathcal{V}_l].
  \label{eq:qpd-estimator}
\end{equation}
where we defined the scaling factor $\gamma:=\sum_l|\gamma_l|\ge 1$, the positive coefficients $p_l:=|\gamma_l|/\gamma$ (such that the set $\{p_l\}$ is a valid probability distribution) and the signed weights $\xi_l := \gamma\,\mathrm{sgn}(\gamma_l)$.~\Cref{eq:qpd-estimator} shows that $\xi_l \mathcal{V}_l$ is a \emph{unbiased estimator} of the channel $\mathcal{U}$.
Then, due to linearity of the trace, we have that
\begin{equation}
    \langle O \rangle:=\mathrm{Tr}[O\,\mathcal{U}(\rho_0)] = \mathbb{E}_l[\xi_l \,\mathrm{Tr}[O\,\mathcal{V}_l(\rho_0)]].
    \label{eq:qpd-estimator-expect}
\end{equation}
Hence, by evolving the initial state $\rho_0$ through a channel $\mathcal{V}_l$ sampled according to the distribution $\{p_l\}$, performing a measurement in the eigenbasis of $O$ and rescaling its outcome by $\xi_l$ (the \emph{post-processing} mentioned in the main text), one obtains an unbiased estimate $\hat{O}$ of $\langle O \rangle$. 
(In practice this estimate is averaged over $M$ repetitions.) 
The terminology \emph{unbiased Hamiltonian simulation}, used in the main text, refers to this property~\eqref{eq:qpd-estimator-expect}, or equivalently to~\cref{eq:qpd-estimator}.

The rescaling necessary to unbias the estimate introduces some overhead. In particular, given that $|\mathrm{Tr}[O\,\mathcal{V}_l(\rho_0)]|\le\|O\|_\infty$ for any CPTP channel $\mathcal{V}_l$, one obtains a bound on the single-shot variance $\Delta_{\rm 1s}^2$ of the estimate $\widehat O$
\begin{equation}
  \Delta_{\rm 1s}^2 \le \mathbb{E}[\widehat O^2] \le \gamma^2\,\|O\|_\infty^2,
  \label{eq:variance-bound-qpd}
\end{equation}
to be compared to the standard upper bound given by $\|O\|_\infty^2$ obtained when directly evolving the state under $\mathcal{U}$.
Hence, guaranteeing a given variance requires increasing the number of measurements $M$ by a multiplicative factor $\gamma^2>1$, called \emph{sampling overhead}.
When reporting \emph{resources} required to implement a QPD of a channel, we thus report both (i) $\mathbb{E}[N_g]$: the \emph{expected number of gates} entailed (the number of gates to implement each channel $\mathcal{V}_l$, averaged under the distribution $\{p_l\}$) together with (ii) $\gamma^2$: the necessary sampling overhead. Finally, we note that when multiple QPDs, with respective sampling overhead $\gamma(k)^2$, are employed within a given circuit, the resulting construction remains unbiased and that the overall sampling overhead becomes $\gamma^2 = \prod_k (\gamma(k))^2$.
While we have only discussed the estimation of operator expectation value~\eqref{eq:qpd-estimator-expect}, we note that QPD has also been extended to sampling-based algorithms~\cite{liu2025qem, onishi2026weak}.

In what follows we will be concerned with unbiased Hamiltonian simulations for a time-dependent Hamiltonian $H(s)$ that admits a \emph{Hermitian-involution} (h.i.) decomposition
\begin{equation}
  H(s) = \sum_j c_j(s) R_j,
  \label{eq:hi-decomp}
\end{equation}
in which the coefficients $c_j(s)$ are real while the operators $R_j$ are Hermitian involutions.
Recall that the latter are defined as operators satisfying $R_j^\dagger = R_j$ together with $R_j^2 = \mathbb{I}$, and that we defined a $1$-norm accordingly: $\|H(s)\|_{1,{\rm hi}}:=\sum |c_j(s)|$.
A canonical example of basis of h.i. is the standard Pauli basis $\{R_j\}=\{P_j\}$, where each $R_j$ is a Pauli string. In such case we denote the norm as $\|H(s)\|_{1,{\rm Pauli}}$. In addition, for any unitary $V$, the rotated-Pauli family $\{V P_j V^\dagger\}$ is also an h.i. basis and yields an alternative decomposition of $H(s)$.
We stress that h.i. decompositions are not unique, and different choices of h.i. basis can induce different $1$-norm (and thus different implementation resources, as soon discussed).

\subsection{TE-PAI}
\label{app:tepai-algo}
TE-PAI is one instantiation of QPDs for Hamiltonian evolution under a time-dependent Hamiltonian with h.i. decomposition of the form~\Cref{eq:hi-decomp}. Following Ref.~\cite{koczor2024probabilistic},
the time evolution $\mathcal{U}(t)[\rho]$ is first approximated by a first-order Trotter formula (for time-dependent Hamiltonians)
\begin{equation}
  \mathcal{U}(t)[\rho] \approx \prod_{\mu=1}^{N_\tau}\prod_{j}
    \e^{-\im \theta_{j\mu}R_j}
      \,\rho\,
  \bigg(\prod_{\mu=1}^{N_\tau}\prod_{j}
    \e^{-\im \theta_{j\mu}R_j}\bigg)^\dag,
  \label{eq:tepai-trotter}
\end{equation}
with rotation angles $\theta_{j\mu}=c_j(t_\mu)\tau/N_\tau$ for $t_\mu=(\mu-1)\tau/N_\tau$.
Later, we will take the limit $N_\tau\to\infty$, restoring exact evolution.

To implement the time-evolution~\eqref{eq:tepai-trotter}, let us first define the individual channels %
\begin{equation}
    \mathcal{R}_j(\theta)[\rho]:=\e^{-\im\frac{\theta}{2} R_j}\rho\,\e^{\im\frac{\theta}{2} R_j},
\end{equation}
that is, the channel of the rotation gate $\e^{-\im\theta R_j/2}$, so that in particular $\mathcal{R}_j(\pi)[\rho]=R_j\rho R_j$ is the full Pauli conjugation.
TE-PAI proceeds by
 applying the probabilistic angle interpolation (PAI)~\cite{koczor2024probabilistic} to each of these unitary channels, such that
\begin{align}
\begin{split}
  \mathcal{R}_j(\theta)
  &= \gamma_1\mathcal{I}
  + \gamma_2\mathcal{R}_j(\mathrm{sgn}(\theta)\Delta)
  + \gamma_3\mathcal{R}_j(\pi)
  \\
  &= \gamma\left(
      p_1\,\mathrm{sgn}(\gamma_1)\mathcal{I}
      + p_2\,\mathrm{sgn}(\gamma_2)\mathcal{R}_j(\mathrm{sgn}(\theta)\Delta)
      + p_3\,\mathrm{sgn}(\gamma_3)\mathcal{R}_j(\pi)
  \right),
\end{split}
  \label{eq:PAI-decomp}
\end{align}
with $\mathcal{I}$ denoting the identity channel and where we introduced the coefficients
\begin{equation}
\begin{cases}
\gamma_1=(1+\cos\theta)/2-(|\sin\theta|/2)\cot(\Delta/2),\\
\gamma_2=|\sin\theta|/\sin\Delta,\\
\gamma_3=(1-\cos\theta)/2-(|\sin\theta|/2)\tan(\Delta/2),
\end{cases}
\end{equation}
that depend on a free parameter $\Delta\in(0,\pi)$.
Note that the coefficients satisfy $\gamma_1+\gamma_2+\gamma_3=1$ but can be negative.
\Cref{eq:PAI-decomp} is a QPD~\eqref{eq:qpd-abstract} specialised to a single rotation $\mathcal{R}_j$, with probabilities 
\begin{align}
\label{eq:tepai_probability}
p_l:=|\gamma_l|/\gamma \quad\text{ and }\quad\gamma:=\sum_{l'}|\gamma_{l'}|>1.
\end{align}
It involves three channels, the identity $\mathcal{I}$ and two Pauli rotations $\mathcal{R}_j$, with distinct angles, generated by $R_j$. The sampled channel (either the identity or one of the two rotations), weighted by the sign $\mathrm{sgn}(\gamma_l)$, recovers $\mathcal{R}_j(\theta)$ in expectation up to the normalisation factor $\gamma$, which contributes to a sampling overhead per rotation of $\gamma^2$. 

The QPD of~\cref{eq:PAI-decomp} is applied to each of the Pauli rotations occurring in~\cref{eq:tepai-trotter}.
Substituting
$\theta=2\theta_{j\mu}=2c_j(t_\mu)\,\tau/N_\tau$
into \cref{eq:PAI-decomp} and summing over $j$, we find that the probability of selecting a non-identity channel at the $\mu$-th time slice is
\begin{equation}
  \frac{\tau\sum_j|c_j(t_\mu)|}{N_\tau}\,\frac{3-\cos\Delta}{\sin\Delta}
  + \mathcal{O}\!\left(\frac{\tau^2}{N_\tau^2}\right),
  \label{eq:tepai-gate-each-pauli}
\end{equation}
with scaling factor,
\begin{equation}
  \gamma(t_\mu)
  :=
  1
  +
  \frac{2\tau\sum_j|c_j(t_\mu)|}{N_\tau}\tan\frac{\Delta}{2}
  +
  \mathcal{O}\!\left(\frac{\tau^2}{N_\tau^2}\right).
  \label{eq:tepai-normalization-each-pauli}
\end{equation}

The approximation bias in \cref{eq:tepai-trotter} can be eliminated by taking the $N_\tau\to\infty$ limit, where a non-identity unitary channel is inserted according to a Poisson process with an inhomogeneous (time-dependent) rate
\begin{align}
\label{eq:poisson_rate_tepai}
    \frac{3-\cos\Delta}{\sin\Delta}\sum_j|c_j(t)|.
\end{align}
When performing time evolution for a time $\tau$ this gives the expected number of non-identity channels per circuit:
\begin{equation}
  \mathbb{E}[N_{\mathrm{g}}]
  = \frac{3-\cos\Delta}{\sin\Delta}
    \lambda \quad \textrm{with} \quad \lambda := \sum_j\int_0^\tau |c_j(s)|\,\diff s,
  \label{eq:tepai-Ngates-derived}
\end{equation}
where we integrated the (gate insertion) rate~\eqref{eq:poisson_rate_tepai} over $t\in[0,\tau]$ and summed over the different Hamiltonian terms.

\subsection{Implementation}
\label{app:tepai-impl}

We first draw an integer $N_{\mathrm{g}}$ from a Poisson
distribution with the expected number of events given by~\cref{eq:tepai-Ngates-derived}. The integer $N_{\mathrm{g}}$ indicates the number of non-identity channels (i.e.\ the number of rotations, or gate count) in a circuit. For each insertion $m\in \{1,2,\dots,N_{\mathrm{g}}\}$:
\begin{enumerate}
\item Draw a time $t_m\in[0,\tau]$ from the density $p(t) = \sum_j |c_j(t)|/\lambda$~\eqref{eq:tepai-Ngates-derived}.
\item Draw a label $j$ of a Hamiltonian term with probability (conditioned on $t_m$) $p(j)=|c_j(t_m)|/\sum_{j'}|c_{j'}(t_m)|$.
\item Pick $\mathcal{R}_j(\mathrm{sgn}(c_j(t_m))\Delta)$ with probability $p_2/(p_2+p_3)$~\eqref{eq:tepai_probability} and $\mathcal{R}_j(\pi)$ otherwise, and retain the sign $\varepsilon_m=\mathrm{sgn}(\gamma_l)$ of the corresponding coefficient in~\eqref{eq:PAI-decomp}.
\end{enumerate}
Apply the sampled channels in increasing order of
sampled times $t_m$ to obtain one circuit instance, of overall sign $\varepsilon=\prod_{m=1}^{N_{\mathrm{g}}}\varepsilon_m$.
Following the quasi-probabilistic estimator~\eqref{eq:qpd-estimator-expect}, the value of the observable $O$ measured on the output state is multiplied by the sign $\varepsilon$ and rescaled by the per-circuit sampling-overhead factor $\gamma_\infty(\tau)$~\eqref{eq:gamma-Q1} to provide the unbiased estimate of $\langle O\rangle$.

\subsection{Resources: per-step gate count and \texorpdfstring{$r$}{r}-step composition}
\label{app:tepai-resources}

We now collect the resources used by TE-PAI to implement a single step of unitary evolution of length $\tau$, then describe how they compose when $r$ such primitives are employed in a given circuit.
Already, the expected number of non-identity channels (i.e., the gate count) per circuit was provided in~\cref{eq:tepai-Ngates-derived}. Because this count is Poisson distributed, its standard deviation is the square root of its mean, $\sqrt{\mathbb{E}[N_{\mathrm{g}}]}$.
the sampling overhead is obtained as
\begin{equation}
  \gamma_\infty(\tau)^2
  := \lim_{N_\tau\to\infty}\prod_{\mu=1}^{N_\tau}\gamma(t_\mu)^2
  = \exp\!\bigg[4\tan\!\Big(\frac{\Delta}{2}\Big)
       \sum_j\int_0^\tau |c_j(s)|\,\diff s\bigg].
  \label{eq:tepai-gamma-derived}
\end{equation}

As per our previous discussion, when $r$ TE-PAI primitives are employed within one circuit (e.g., one per Trotter step), the composed channel is unbiased, and the overall sampling overhead is the product of the per-step overheads:
\begin{equation}
\begin{aligned}
  \gamma^2_{\mathrm{tot}}
  &= \prod_{\mathrm{seg}=1}^{r}\gamma_\infty(\tau)^2_{\mathrm{seg}}
  =: \e^{V}, \\
  V &:= \sum_{\mathrm{seg}=1}^{r} V_{\mathrm{seg}},
\end{aligned}
  \label{eq:gamma-composed}
\end{equation}
where $V_{\mathrm{seg}}:=\log\gamma_\infty(\tau)^2_{\mathrm{seg}}$ is the per-step log-overhead. A uniform per-step budget $V_{\mathrm{seg}}=V/r$ achieves the global budget $\e^V$, which is the prescription ``replace $V\to V/r$'' of Lemma~1 in the main text.
Because each per-step estimator carries a signed weight of constant magnitude $\gamma_\infty(\tau)$, the composed weight over the $r$ steps has constant magnitude $\prod_{\mathrm{seg}}\gamma_\infty(\tau)_{\mathrm{seg}}=\e^{V/2}$, which propagates the single-step variance bound $\Delta_{\rm 1s}^2\le\e^{V}\lVert O\rVert_\infty^2$ to the full $r$-step circuit.
We summarise the discussion above in the following lemma, that builds upon the previous definitions.

\begin{lemma}[TE-PAI: channel-level unbiasedness and gate count]
\label{lem:tepai-app}
Let $H(s)$ be a time-dependent Hamiltonian and let $\mathcal{U}(\tau)$ denote the channel generated by
$H$ for an evolution time $\tau$.
Given an h.i. decomposition $H(s)=\sum_j c_j(s)\,R_j$ 
and $\Delta\in(0,\pi)$ define the integrated rate
\begin{equation}\label{def:rate_app}
\lambda := \sum_j\int_0^\tau |c_j(s)|\,\diff s.
\end{equation}
One can sample a physical unitary channel $\mathcal{V}_l$ and a sign
$\varepsilon_l\in\{\pm 1\}$ such that, with the signed weight
$\xi_l := \gamma_\infty(\tau)\,\varepsilon_l$, the pair
$(\mathcal{V}_l,\xi_l)$ implements unbiased Hamiltonian simulation (as a channel)
in the sense of \cref{eq:qpd-estimator}:
\begin{equation}
  \mathbb{E}_l\bigl[\xi_l\,\mathcal{V}_l\bigr] = \mathcal{U}(\tau).
\end{equation}
The construction requires an expected gate count
\begin{equation}
  \mathbb{E}[N_{\mathrm{g}}]
  =
  \frac{(3-\cos\Delta)\lambda}{\sin\Delta}
  \label{eq:N_non-id-Q1}
\end{equation}
per circuit (with standard deviation of this gate-count given by $\sqrt{\mathbb{E}[N_{\mathrm{g}}]}$) together with a sampling overhead
\begin{equation}
  \gamma_\infty(\tau)^2=\e^{4\lambda\tan(\Delta/2)}.
  \label{eq:gamma-Q1}
\end{equation}
\end{lemma}

Having reviewed TE-PAI, we are ready to provide a proof of Lemma~1.
\begin{proof}[Proof of Lemma~1]
Provided a sampling overhead $\gamma_\infty(\tau)^2=\e^V$ for a fixed $V>0$, take $\Delta^\ast = 2\arctan(V/(4\lambda))$.
For simulation of $H(s)$ for a time $\tau$, substituting $\Delta=\Delta^\ast$ into~\cref{eq:N_non-id-Q1} gives
\begin{equation}
  \mathbb{E}[N_{\mathrm{g}}]
  = \frac{4\lambda^2}{V} + \frac{V}{2},
  \label{eq:tepai-cost-app}
\end{equation}
which is~Eq.~(2) of Lemma~1. 
As previously discussed around~\cref{eq:gamma-composed}, when simulating $r$ independent steps of evolution within the same circuit, to ensure an overall sampling of $e^V$, replace $V\rightarrow V/r$ to get the expected circuit depth per step.
\end{proof}

\section{Gate counts for PTER}
\label{app:apply-tepai}

This appendix combines results from~\cref{app:trotter_error,app:tepai-details} to derive gate counts for $\PT$. We use the nested-commutator representation of $G_k$ (\cref{app:higher-order}) to assess the number of gates needed to implement the correction channels through TE-PAI. This allows us to obtain the number of gates required by PTER at a fixed number of steps, and we then optimise it to minimise the overall gate count. 
\Cref{app:rate-parameter} translates the kernel form of $G_k$ into the Pauli-basis integrated rate $\lambda_k(\tau)$ that controls the resources required to implement the correction channel through TE-PAI.
\Cref{app:gate-count-unified} details, for geometrically local Hamiltonians, the gate count per step of PTER together with the optimisation of the number of steps, proving Theorems~1 and~2 of the main text.
\Cref{app:non-geometric-gate-count} extends the analysis to general (non-local) Pauli Hamiltonians.

\subsection{Pauli-basis integrated rate}
\label{app:rate-parameter}

The resources required by TE-PAI (\cref{lem:tepai-app}) are controlled by the integrated rate of~\cref{def:rate_app}, which is the integrated
$1$-norm of the Hamiltonian in the chosen h.i. decomposition.
To implement the correction channel, we use the standard Pauli basis to decompose the
remainder Hamiltonian $G_k(s)$, so that the integrated rate $\lambda_k$ is defined as
\begin{equation}
  \lambda_k(\tau)
  := \int_0^\tau \|G_k(s)\|_{1,\mathrm{Pauli}}\,\diff s.
  \label{eq:lambda-k-def}
\end{equation}
This subsection identifies the small-$\tau$ expansion of
$\lambda_k(\tau)$ that will be used in the optimisation of
\cref{app:gate-count-unified}.
We start with the kernel form of $G_k$ derived in \cref{eq:Gk-remainder-explicit},
\begin{equation}
  G_k(s)
  = \sum_\nu\int_0^s
    K_\nu(s,\sigma)\,
    \mathrm{Ad}_{L_k^{j_\nu}(s,\sigma)}(C_\nu)
    \,\diff\sigma,
  \label{eq:Gk-remainder-explicit-rec}
\end{equation}
with the multi-index $\nu$ defined in~\cref{eq:nu-label}, scalar prefactor $K_\nu(s,\sigma)$ given by \cref{eq:Knu-explicit}, and $C_\nu$ given by \cref{eq:Cnu-explicit}. %
The kernel weight $w_\nu$ and its normalisation $\int_0^s K_\nu(s,\sigma)\,\diff\sigma = s^k w_\nu$ are those of \cref{eq:wnu-int,eq:wnu-def}.
Let us further define the following coefficients:
\begin{align}
  \alpha_k := \Bigl\|\sum_\nu w_\nu\,C_\nu\Bigr\|_{1,\mathrm{Pauli}}
  \qquad \text{and} \qquad
  \tilde\alpha_k := \sum_\nu w_\nu\,\|C_\nu\|_{1,\mathrm{Pauli}}.
  \label{eq:alpha-tilde-alpha-def}
\end{align}
By the triangle inequality, $\alpha_k\le \tilde\alpha_k$.  %

\begin{lemma}[Integrated rate $\lambda_k$ for PTER]
\label{lem:rate-parameter}
Let $\mathcal{S}_k$ be a fixed $k$-th order Trotter formula applied to
a geometrically local Hamiltonian (Sec.~A1), and let $\alpha_k,\tilde\alpha_k$ be defined by
\cref{eq:alpha-tilde-alpha-def}.  Then, as $\tau\to0$,
\begin{equation}
  \lambda_k(\tau)
  = \frac{\alpha_k\,\tau^{k+1}}{k+1}
    + \mathcal{O}\!\left(\tilde\alpha_k\,\tau^{k+2}\right),
  \label{eq:rate-parameter-additive}
\end{equation}
with a constant depending only on $k$, the formula, and the locality parameters of Sec.~A1.
\end{lemma}

\begin{proof}
Recall from \cref{eq:Lk-two-arg} that
$L_k^{j_\nu}(s,\sigma)$ is a partial product of rotations
$\e^{-\im s\widetilde H_j}$ (with $j <j_{\nu}$) and a rotation $\e^{-\im \sigma \widetilde H_{j_{\nu}}}$. Since each $\widetilde H_j$ is a weighted sum of mutually commuting Pauli
strings, this partial product factorises into a product of single-Pauli rotations.
We use the elementary conjugation rule
\begin{equation}
  \e^{-\im\theta P}\,R\,\e^{\im\theta P}
  = \begin{cases}
    R, & [P,R]=0, \\[2pt]
    \cos(2\theta)R-\im\sin(2\theta)PR, & \{P,R\}=0.
  \end{cases}
  \label{eq:pauli-conj-rule}
\end{equation}
for Pauli strings $P$ and $R$.  %

As a first step, we wish to bound $\|G_k(s)\|_{1,\mathrm{Pauli}}$. 
The leading-order contribution in $s$ to
\cref{eq:Gk-remainder-explicit-rec}
is
\begin{align}
  G_k^{(0)}(s) := s^k\sum_\nu w_\nu\,C_\nu
  \quad \text{such that} \quad
  \|G_k^{(0)}(s)\|_{1,\mathrm{Pauli}} = s^k\,\alpha_k.
  \label{eq:Gk-bare}
\end{align}
We now show that
\begin{equation}
  \|G_k(s)\|_{1,\mathrm{Pauli}}
  = s^k\,\alpha_k+\mathcal{O}\!\left(\tilde\alpha_k\,s^{k+1}\right).
  \label{eq:Gk-norm-asymp}
\end{equation}
First, note that the triangle inequality gives
\begin{equation}
\bigl|\,\|G_k(s)\|_{1,\mathrm{Pauli}} - s^k\alpha_k\,\bigr|
  = \bigl|\,\|G_k(s)\|_{1,\mathrm{Pauli}} -\|G_k^{(0)}(s)\|_{1,\mathrm{Pauli}}\,\bigr| \le \|G_k(s)-G_k^{(0)}(s)\|_{1,\mathrm{Pauli}},
  \label{eq:reverse-triangle}
\end{equation}
so it remains to bound the Pauli $1$-norm of the difference between $G_k(s)$ and
$G_k^{(0)}(s)$.  Subtracting \cref{eq:Gk-bare} from the kernel form~\cref{eq:Gk-remainder-explicit-rec} gives
\begin{equation}
  G_k(s)-G_k^{(0)}(s)
  =\sum_\nu\int_0^s K_\nu(s,\sigma)\,\bigl[\mathrm{Ad}_{L_k^{j_\nu}(s,\sigma)}(C_\nu)-C_\nu\bigr]\,\diff\sigma.
  \label{eq:Gk-remainder-form}
\end{equation}

We next bound each integrand in Pauli $1$-norm.
To do so, expand
$C_\nu=\sum_R c_R R$ in the Pauli basis and note that
\begin{equation}\label{eq:part0}
  \|\mathrm{Ad}_{L_k^{j_\nu}(s,\sigma)}(C_\nu)-C_\nu\|_{1,\mathrm{Pauli}}
  \le \sum_R |c_R|\,\|\mathrm{Ad}_{L_k^{j_\nu}(s,\sigma)}(R)-R\|_{1,\mathrm{Pauli}}.
\end{equation}
Let
$\mathcal{A}_\eta:=\e^{-\im\theta_\eta \mathrm{ad}_{P_\eta}}$, then from
\cref{eq:pauli-conj-rule}, for every Pauli string $R'$ that \emph{anti-commutes} with $P_\eta$, we have
\begin{align}
\label{eq:channel_inequalities}
\begin{split}
  &\|(\mathcal{A}_\eta-I)(R')\|_{1,\mathrm{Pauli}}
  =
  |\cos(2\theta_\eta)-1| + |\sin(2\theta_\eta)|
  =
  2\sin^2(\theta_\eta) + |\sin(2\theta_\eta)|
  \le
  4|\theta_\eta|,
  \\
  &\|\mathcal{A}_\eta(R')\|_{1,\mathrm{Pauli}}
  =
  |\cos(2\theta_\eta)| + |\sin(2\theta_\eta)|
  \le
  (1+2|\theta_\eta|).
\end{split}
\end{align}
For a Pauli string $R'$ that \emph{commutes} with $P_\eta$, we rather get 
\begin{align}
\begin{split}
&\|(\mathcal{A}_\eta-I)(R')\|_{1,\mathrm{Pauli}}=0\\
&\|\mathcal{A}_\eta(R')\|_{1,\mathrm{Pauli}}=1.
\end{split}
\end{align}
Hence, the bounds~\eqref{eq:channel_inequalities} hold for any Pauli string $R'$.
By linearity and the triangle inequality, the same bounds extend to any Pauli
expansion $X=\sum_{R'}x_{R'}R'$:
\begin{equation}
  \|(\mathcal{A}_\eta-I)(X)\|_{1,\mathrm{Pauli}}
  \le 4|\theta_\eta|\,\|X\|_{1,\mathrm{Pauli}},
  \qquad
  \|\mathcal{A}_\eta(X)\|_{1,\mathrm{Pauli}}
  \le (1+2|\theta_\eta|)\,\|X\|_{1,\mathrm{Pauli}}.
  \label{eq:single-rotation-1norm-bounds}
\end{equation}

To proceed further, let us now use the geometric locality of $H$ (Sec.~A1). A Pauli string $R$ appearing in $C_\nu$ arises from a non-vanishing nested commutator of Hamiltonian terms. Since disjoint Pauli strings commute, these terms must overlap, so $R$ is supported in a ball of constant radius, fixed by $k$ and the locality range. Conjugation by the partial product $L_k^{j_\nu}(s,\sigma)$ can enlarge this support only by a further constant amount. Let $\Omega(R)$ be the set of single-Pauli rotations in $L_k^{j_\nu}(s,\sigma)$ whose support intersects this enlarged ball. Geometric locality then guarantees $|\Omega(R)|=\mathcal{O}(1)$ uniformly in the system size. Every rotation outside $\Omega(R)$ commutes with every Pauli string generated while conjugating $R$, and may be removed from the product.
Index the elements of $\Omega(R)$ as $\mathcal{A}_1,\ldots,\mathcal{A}_{M_R}$ with
$M_R=\mathcal{O}(1)$, where $\mathcal{A}_\eta=\e^{-\im\theta_\eta\mathrm{ad}_{P_\eta}}$.
Each active angle $\theta_\eta$ is a product of a Suzuki coefficient, a Pauli coefficient of $H$, which is bounded by assumption, and either $s$ or $\sigma\le s$.  Hence
\begin{equation}
  \sum_{\eta=1}^{M_R}|\theta_\eta|=\mathcal{O}(s)
  \label{eq:active-angle-sum}
\end{equation}
uniformly in $\sigma$ and independently of the system size.
Using the telescoping identity
\begin{equation*}
  \mathcal{A}_{M_R}\cdots\mathcal{A}_1-I
  =
  \sum_{\eta=1}^{M_R}
  (\mathcal{A}_\eta-I)
  \mathcal{A}_{\eta-1}\cdots\mathcal{A}_1
\end{equation*}
and \cref{eq:single-rotation-1norm-bounds}, we get
\begin{equation}
\begin{aligned}
  \|\mathrm{Ad}_{L_k^{j_\nu}(s,\sigma)}(R)-R\|_{1,\mathrm{Pauli}}
  &\le
  \sum_{\eta=1}^{M_R}
  4|\theta_\eta|
  \prod_{\eta'<\eta}(1+2|\theta_{\eta'}|)
  \le
  4\Bigl(\sum_{\eta=1}^{M_R}|\theta_\eta|\Bigr)
  \exp\!\Bigl(2\sum_{\eta=1}^{M_R}|\theta_\eta|\Bigr)
  =
  \mathcal{O}(s).
  \label{eq:per-term-deviation-oneP}
\end{aligned}
\end{equation}
Hence from the previous and \cref{eq:part0}, we obtain
\begin{equation}
  \|\mathrm{Ad}_{L_k^{j_\nu}(s,\sigma)}(C_\nu)-C_\nu\|_{1,\mathrm{Pauli}}
  = \mathcal{O}\left(s\right) \|C_\nu\|_{1,\mathrm{Pauli}}.
  \label{eq:per-term-deviation}
\end{equation}

Applying the triangle inequality to \cref{eq:Gk-remainder-form}, using
$K_\nu(s,\sigma)\ge0$, and then using \cref{eq:wnu-int}, we obtain
\begin{equation}
\begin{aligned}
  \|G_k(s)-G_k^{(0)}(s)\|_{1,\mathrm{Pauli}}
  &\le \sum_\nu\int_0^s K_\nu(s,\sigma)\,
      \|\mathrm{Ad}_{L_k^{j_\nu}(s,\sigma)}(C_\nu)-C_\nu\|_{1,\mathrm{Pauli}}
      \,\diff\sigma \\
  &= \mathcal{O}(s)
      \sum_\nu\|C_\nu\|_{1,\mathrm{Pauli}}\int_0^sK_\nu(s,\sigma)\,\diff\sigma 
      \\
  &= \mathcal{O}\!\left(\tilde\alpha_k\,s^{k+1}\right).
\end{aligned}
  \label{eq:Gk-deviation}
\end{equation}
Together with \cref{eq:reverse-triangle}, this proves
\cref{eq:Gk-norm-asymp}.  Integrating from $0$ to $\tau$ gives
\cref{eq:rate-parameter-additive}.
\end{proof}

\subsection{Gate-count optimisation}\label{app:gate-count-unified}
This subsection proves Theorems~1 and~2 of the main text and relies
on the scaling of the integrated rate $\lambda_k(\tau)$ obtained in the previous subsection.
Let us first evaluate the gate count for $\PT_k$ at a fixed number of steps $r=t/\tau$.
Implementing $U(t)=U(\tau)^r$ requires $r$ steps of Trotter circuits, interleaved with $r$ steps of TE-PAI corrections.
For the Trotter part, recall that each step incurs $\Upsilon_k N$ gates (\cref{app:trotter-recap}).
For the correction part,
ensuring a per-step sampling overhead $\gamma_\infty(\tau)^2=\e^{V/r}$ (so the total overhead $\gamma^2_{\mathrm{tot}}=\e^V$ over $r$ steps) incurs an expected number of gates $4\lambda_k(\tau)^2/(V/r)+(V/r)/2$ (\cref{app:tepai-resources}). Hence the total expected gate count is given by
\begin{equation}
  \mathbb{E}[N_{\mathrm{g}}](r)
  = r\,\Upsilon_k N
    + r\!\left(\frac{4\lambda_k(t/r)^2}{V/r} + \frac{V}{2r}\right)
  = r\,\Upsilon_k N + \frac{4\,r^2\,\lambda_k(t/r)^2}{V} + \frac{V}{2}.
  \label{eq:gates-r-unified}
\end{equation}

With $\tau=t/r$ and $\lambda_k'(\tau):=\partial_\tau\lambda_k(\tau)$, differentiating~\cref{eq:gates-r-unified} with respect to $r$, and equating it to $0$, yields the equation
\begin{equation}
  \Upsilon_k N
  = \frac{8\,t\,\lambda_k(\tau)\,\lambda_k'(\tau)
         - 8\,r\,\lambda_k(\tau)^2}{V},
  \label{eq:dN-dr-unified}
\end{equation}
whose unique positive solution is $r^\star$.

To examine how $\mathbb{E}[N_{\mathrm{g}}](r^\star)$ scales in the physical parameters, we use that $\alpha_k=\Theta(n)$. For a geometrically local $H$ on a fixed-dimensional lattice each nested commutator is supported on $\mathcal{O}(1)$ sites, so $\tilde\alpha_k=\Theta(n)$, and $\alpha_k=\Theta(n)=\Theta(\tilde\alpha_k)$ for generic couplings. The leading nested-commutator norm $\alpha_k$ falls below this scale only when the coupling constants are tuned to produce cancellations between the nested commutators, and in such a sublinear regime $\alpha_k=o(n)$ one can expect an even better gate-count scaling. For $\alpha_k=\Theta(n)$, the additive remainder in \cref{eq:rate-parameter-additive} is of relative order $\mathcal{O}((\tilde\alpha_k/\alpha_k)\,\tau)=\mathcal{O}(\tau)$, giving the leading-order $\lambda_k$ and its $\tau$-derivative,
\begin{equation}
\begin{aligned}
    \lambda_k(\tau) &= \frac{\alpha_k\tau^{k+1}}{k+1}\,(1+\mathcal{O}(\tau)), \\
    \lambda_k'(\tau) &= \alpha_k\tau^k\,(1+\mathcal{O}(\tau)).
\end{aligned}
    \label{eq:rate-parameter}
\end{equation}
The expected gate count in \cref{eq:gates-r-unified} then becomes
\begin{equation}
  \mathbb{E}[N_{\mathrm{g}}](r)
  = r\,\Upsilon_k N
    + \frac{4\,\alpha_k^2\,t^{2k+2}}{(k+1)^2\,V\,r^{2k}}\,\bigl(1+\mathcal{O}(t/r)\bigr)
    + \frac{V}{2},
  \label{eq:gates-r-alpha}
\end{equation}
while the equation giving the optimal number of steps in \cref{eq:dN-dr-unified} reduces to
\begin{equation}
  \Upsilon_k N
  = \frac{8k\,\alpha_k^2\,t^{2k+2}}{(k+1)^2\,V\,r^{2k+1}}\,\bigl(1+\mathcal{O}(t/r)\bigr).
  \label{eq:dN-dr-alpha}
\end{equation}
Dropping the $\mathcal{O}(t/r)$ correction gives the closed-form leading-order solution
\begin{equation}
\begin{aligned}
  r^\star_0
  &:= \left(
    \frac{8k\,\alpha_k^2\,t^{2k+2}}{(k+1)^2\,V\,\Upsilon_k N}
    \right)^{\!\!1/(2k+1)} \quad \Rightarrow \quad
  \tau^\star_0
  &:= \frac{t}{r^\star_0}
  = \left(\frac{(k+1)^2\,V\,\Upsilon_k N}{8k\,\alpha_k^2 t}\right)^{\!\!1/(2k+1)}.
\end{aligned}
  \label{eq:r-tau-star-leading}
\end{equation}
The optimal Trotter step is then $\tau^\star_0=\mathcal{O}\big((nt)^{-1/(2k+1)}\big)$, vanishing whenever $nt\gg V\Upsilon_k$.
When $\alpha_k=\Theta(n)$, substituting the leading-order expression for $\lambda_k$ is self-consistent: the $\mathcal{O}(t/r)$ correction in \cref{eq:dN-dr-alpha}, evaluated at $r=r^\star_0$, is exactly $\mathcal{O}(\tau^\star_0)=\mathcal{O}((nt)^{-1/(2k+1)})$, which is $o(1)$ in the limit.
Solving \cref{eq:dN-dr-alpha} self-consistently in this regime,
\begin{equation}
\begin{aligned}
  r^\star &= r^\star_0\,\bigl(1+\mathcal{O}((nt)^{-1/(2k+1)})\bigr), \\
  \tau^\star &= \tau^\star_0\,\bigl(1+\mathcal{O}((nt)^{-1/(2k+1)})\bigr).
\end{aligned}
  \label{eq:r-tau-star-actual}
\end{equation}
Substituting $r^\star$ back into \cref{eq:gates-r-alpha} and using that the second term equals $1/(2k)$ times the first at the leading-order optimum, the gate count $\mathbb{E}[N_{\mathrm{g}}]^\star$ is $(2k+1)/(2k)$ times the Trotter contribution plus $V/2$, giving
\begin{equation}
  \mathbb{E}[N_{\mathrm{g}}]^\star
  = \frac{2k+1}{2k}\,\left(\frac{8k}{(k+1)^2}\right)^{\!\!1/(2k+1)}\,
      V^{-\frac{1}{2k+1}}\,
      (\Upsilon_k N)^{\frac{2k}{2k+1}}\,
      \alpha_k^{\frac{2}{2k+1}}\,
      t^{\frac{2k+2}{2k+1}}\,\bigl(1+o(1)\bigr)
    + \frac{V}{2},
  \label{eq:gates-optimal-unified}
\end{equation}
matching Eq.~(5) of the main text and proving Theorems~1 and~2.
For the numerical comparisons of Fig.~2, the gate count for $\PT$ is obtained by minimising the exact total gate count \cref{eq:gates-r-unified} over the number of steps $r$ directly, without the leading-order substitution, so the plotted counts carry no $o(1)$ error even at the smallest system sizes.

\subsection{Gate count for non-local Hamiltonians}
\label{app:non-geometric-gate-count}

We now drop the geometric-locality assumption of the previous subsection. Let $H=\sum_{\ell=1}^L H_\ell$, where each $H_\ell$ is a weighted sum of mutually commuting Pauli strings, otherwise arbitrary.
The proof of \cref{lem:rate-parameter} used geometric locality to show that, when conjugating any of the Pauli strings appearing in the commutators $C_{\nu}$ by the Pauli rotations composing $L_k^{j_\nu}(\tau,\sigma)$, only a constant number of such rotations contribute.
This was captured by \cref{eq:active-angle-sum}. We cannot rely on this fact any more.

As for \cref{app:rate-parameter}, we start by focusing on an arbitrary Pauli string $R$ appearing in one of the commutators $C_{\nu}$~\eqref{eq:Cnu-explicit}, and consider conjugation by the partial product $L_k^{j_\nu}(\tau,\sigma)$. The latter is decomposed as a product of Pauli rotations, and we index the rotations that contribute nontrivially to the chain of conjugations by $\eta=1, \dots, M_R$. However, we
do not have the guarantee that $M_R \in \mathcal{O}(1)$, as was the case before. Still, we can bound the sum of the rotation angles through
\begin{equation}
  \sum_\eta|\theta_\eta| \le \chi_k\beta\tau,
  \qquad
  \chi_k := \frac{1}{L}\sum_{j=1}^{\tilde\Upsilon_k}|\omega_j|,
  \label{eq:chi-k-def}
\end{equation}
uniformly in $\sigma\in[0,\tau]$, where we recall that the $\{\omega_j\}$ are the Suzuki coefficients from \cref{eq:Sk-product-recap} and $\beta:=\|H\|_{1, \rm{Pauli}}$.
For the symmetric Suzuki formulas used here, every component carries the same total weight $\frac{1}{L}\sum_j|\omega_j|$, which gives the uniform bound above.
Note that at first and second order ($k=1$ and $2$ respectively) these evaluate to $\chi_1=\chi_2=1$. For higher (but fixed) order, these remain constant.

In turn, repeatedly using \cref{eq:channel_inequalities} together with \cref{eq:chi-k-def}, noting that the latter holds for any $R$ of the Pauli decomposition $C_{\nu}$, we obtain
\begin{equation}
  \bigl\|\mathrm{Ad}_{L_k^{j_\nu}(\tau,\sigma)}(C_\nu)\bigr\|_{1,\mathrm{Pauli}}
  \le \|C_\nu\|_{1,\mathrm{Pauli}}\,\prod_\eta\bigl(1+2|\theta_\eta|\bigr)
  \le \|C_\nu\|_{1,\mathrm{Pauli}}\,\e^{2\chi_k\beta\tau}.
  \label{eq:conj-pauli-1norm-k}
\end{equation}
Combining \cref{eq:conj-pauli-1norm-k} with the kernel form of Lemma~2, we then get
\begin{equation}
  \|G_k(s)\|_{1,\mathrm{Pauli}}
  \le \sum_\nu\int_0^s K_\nu(s,\sigma)\,\|C_\nu\|_{1,\mathrm{Pauli}}\,\e^{2\chi_k\beta s}\,\diff\sigma
  = \tilde\alpha_k\,s^k\,\e^{2\chi_k\beta s}.
  \label{eq:Gk-norm-nongeom}
\end{equation}
It follows that when integrating it over $s\in[0,\tau]$, and noting that $\e^{2\chi_k\beta s}\le\e^{2\chi_k\beta\tau}$, we obtain
\begin{equation}
  \lambda_k
=\int_0^\tau\|G_k(s)\|_{1,\mathrm{Pauli}}\,\diff s
  \le \tilde\alpha_k\,\e^{2\chi_k\beta\tau}\,\frac{\tau^{k+1}}{k+1}.
  \label{eq:lambda-nongeom}
\end{equation}
Note that, compared to the leading term of \cref{eq:rate-parameter-additive}, the upper bound has been multiplied by a factor $\e^{2\chi_k\beta\tau}>1$ and the coefficient $\tilde\alpha_k$ now replaces the tighter $\alpha_k$ of \cref{lem:rate-parameter}.

With the integrated rate $\lambda_k$ now bounded~\eqref{eq:lambda-nongeom}, we can repeat the same steps as for the geometric case of \cref{app:gate-count-unified}. First, inserting \cref{eq:lambda-nongeom} into \cref{eq:gates-r-alpha} with the substitution $\alpha_k \to g_k := \tilde\alpha_k\,\e^{2\chi_k\beta\tau}$ gives
\begin{equation}
  \mathbb{E}[N_{\mathrm{g}}]
  \le r\,\Upsilon_k N
    + \frac{4\,g_k^{2}\,t^{2k+2}}{(k+1)^2\,V\,r^{2k}}
    + \frac{V}{2}.
\label{eq:gates-nongeom-k}
\end{equation}
We now identify the number of steps $r^\star$ that minimises \cref{eq:gates-nongeom-k}.
Specialising \cref{eq:r-tau-star-leading} under the substitution $\alpha_k\to g_k$, and using that $\tilde\alpha_k=\mathcal{O}(\beta^{k+1})$, we obtain
\begin{equation}
\begin{aligned}
  r^{\star}
  &=\Theta\!\left(V^{-1/(2k+1)}\,(\beta\,t)^{(2k+2)/(2k+1)} N^{-1/{(2k+1)}}\right), \\
  \tau^{\star}=t/r^\star
  &=\Theta\!\left(V^{1/(2k+1)}\beta^{-(2k+2)/(2k+1)}t^{-1/(2k+1)}N^{1/(2k+1)}\right).
\end{aligned}
  \label{eq:rstar-nongeom}
\end{equation}
Note that the term $\chi_k\beta\tau^{\star}=\mathcal{O}\bigl(V^{1/(2k+1)}t^{-1/(2k+1)}\bigr)$ goes towards zero at large $t$ for fixed $V$, so we have $\e^{2\chi_k\beta\tau^\star}=1+o(1)$ and $g_k=(1+o(1))\,\tilde\alpha_k$.
Then, substituting $r^\star$ back into \cref{eq:gates-nongeom-k} gives
\begin{equation}
  \mathbb{E}[N_{\mathrm{g}}]
  = \mathcal{O}\!\left(
      V^{-1/(2k+1)}
      N^{2k/{(2k+1)}}\beta^{(2k+2)/{(2k+1)}}\,t^{1+1/(2k+1)}
    \right).
  \label{eq:gates-nongeom-final-k}
\end{equation}
Since $\beta = \mathcal{O}(N)$, it reproduces the results stated in Eq.~(6). 
For $t \geq V$ the exponential prefactor is $\mathcal{O}(1)$, so \cref{eq:gates-nongeom-final-k} is accurate as a non-asymptotic upper bound throughout the parameter range of interest.
At constant sampling overhead this reduces to $\mathcal{O}(N\beta\,t^{1+1/(2k+1)})$:
the $k=1$ bound is $\mathcal{O}(N\beta\,t^{4/3})$ and the $k=2$ bound $\mathcal{O}(N\beta\,t^{6/5})$.

\section{Sampling the correction channel}
\label{app:sampling}

In the previous section, we derived gate-count scalings for PTER leaving aside the question of how it is implemented.
In this section, we detail how TE-PAI (\cref{app:tepai-details}) is implemented to simulate the time-dependent evolution under the remainder Hamiltonian $G_k(s)$ for $s\in [0, \tau]$.
\Cref{app:meth} recalls the quantities that are required for the implementation of TE-PAI, including an h.i.\ decomposition of $G_k(s)$ and the associated integrated rate.
\Cref{app:sampling-geom} treats the geometrically local case, where $G_k$ admits a straightforward explicit Pauli decomposition.
\Cref{app:nongeom-sampler} treats the non-local case, where we need to resort to an alternative h.i.\ decomposition. This incurs a slight overhead but leaves the gate-count scalings (derived in \cref{app:non-geometric-gate-count}) unchanged.

\subsection{Preamble}\label{app:meth}
The implementation items of TE-PAI presented in \cref{app:tepai-impl} were based on knowledge of a Hermitian-involution (h.i.)
decomposition of the remainder Hamiltonian~\eqref{eq:hi-decomp} such as
\begin{equation}
  G_k(s) = \sum_j c_j(s)\,R_j,
  \label{eq:G_gen}
\end{equation}
with $\{R_j\}$ a set of h.i. operators (for instance, Pauli strings).
Given such a decomposition, one obtains the $1$-norm of the Hamiltonian at any time $s$,
\begin{equation}\label{eq:app_hi_norm}
\| G_k(s)\|_{1, \mathrm{hi}} := \sum_j|c_j(s)|,
\end{equation}
from which can be evaluated the integrated rate 
\begin{equation}\label{eq:app_int_rate}
\lambda_k := \int_0^\tau\| G_k(s)\|_{1, \mathrm{hi}}\,\diff s.
\end{equation}
\cref{eq:app_int_rate,eq:app_hi_norm} are the quantities needed to implement TE-PAI.  
However, the form of $G_k(s)$ obtained so far is only a nested-commutator kernel form~\eqref{eq:Gk-remainder-explicit}, impeding the evaluation of \cref{eq:app_hi_norm,eq:app_int_rate}.
As we shall see, for the case of geometrically local Hamiltonians, conversion from the kernel form to a Pauli decomposition~\eqref{eq:G_gen} is straightforward (\cref{app:sampling-geom}). However, for the case of non-local Hamiltonians additional care is needed (\cref{app:nongeom-sampler}). In particular, for the latter we will exploit the fact that h.i. decompositions are not unique. 

For both the geometrically local and non-local cases, before discussing general situations, we start by dealing with a \emph{simplified setting} where $k=1$ and $H=A+B$. Here $A$ (or $B$) consists of $N_A$ (or $N_B$) mutually commuting Pauli strings such that $N=N_A+N_B$ is the number of distinct Pauli terms in $H$. 
In this case, we recall that the remainder Hamiltonian is given by~\eqref{eq:G1-kernel-app}
\begin{equation}
  G_1(s) = \sum_h c_h \int_0^s
    \e^{-\im As}\,\e^{-\im B\sigma}\,P_h\,\e^{\im B\sigma}\,\e^{\im As}\,\diff\sigma 
    = \sum_{h} c_h \int_0^s\e^{-\im s\,\mathrm{ad}_A}\,\e^{-\im\sigma\,\mathrm{ad}_B}(P_h)\,\diff\sigma,
  \label{eq:G1-recall}
\end{equation}
with coefficients $c_h$ stemming from the Pauli decomposition of the commutator $[A,B]= \im\sum_h c_h P_h$, from which we defined 
\begin{equation}\label{eq:def_app_alpha1}
    \alpha_1:=\sum_h|c_h|.
\end{equation}
For the second equality in \cref{eq:G1-recall} we used that $\e^{- \im \alpha \mathrm{ad}_X}$ is the conjugation by the unitary $\e^{- \im \alpha X}$. Since $A$ and $B$ are each a sum of mutually commuting Pauli strings, 
$\e^{-\im s\,\mathrm{ad}_A}\e^{-\im\sigma\,\mathrm{ad}_B}$ factors into conjugations by $N=N_A+N_B$ single-Pauli rotations. 
We index these
rotations through $\eta=1,\dots,N$ from right to left, and write $P_\eta$ for the Pauli string generating the
rotation $\eta$ and $\theta_\eta$ for its angle. 
With this notation, we have
\begin{equation}\label{eq:app_one_conjug}
    \e^{-\im s\,\mathrm{ad}_A}\,\e^{-\im\sigma\,\mathrm{ad}_B}(P_h) = %
    \e^{-\im\theta_{N}\,\mathrm{ad}_{P_N}} \cdots \e^{-\im\theta_{1}\,\mathrm{ad}_{P_1}}(P_h).
\end{equation}
Furthermore, the values of the angles involved are given explicitly by
\begin{equation}\label{app:angles_def}
    \theta_{\eta} = c_\eta \omega_\eta t_{\eta} \qquad \textrm{with} \quad t_{\eta} = \begin{cases}
        \sigma & (\eta\leq N_B) \\ 
        s & (\eta>N_B)
    \end{cases},
\end{equation}
$c_{\eta}$ being the weight of $P_{\eta}$ in the Pauli decomposition of $H$ and $\omega_{\eta}$ being the corresponding Suzuki coefficients (with $\omega_{\eta}=1$ for $k=1$).

\subsection{Geometrically local case: explicit Pauli decomposition}
\label{app:sampling-geom}

For the geometrically local case, recasting the generator~\eqref{eq:G1-recall} into an explicit Pauli decomposition is relatively straightforward. 
In \cref{app:impl-local-ppp}, we detail how we obtain the explicit Pauli decomposition of $G_1(s)$ that can be used for the implementation of TE-PAI (and thus $\PT$) for the simplified setting. To do so, we resort to \emph{Pauli path propagation}, which is detailed as it will also form the basis of the treatment of the non-local case later on. 
Similar propagations are employed in approximate classical circuit simulation techniques in~\cite{Aharonov2022, Fontana2023}, where some of the paths are truncated (no truncation is performed here).
This is then generalised to higher order and a general splitting of the Hamiltonian in \cref{app:impl-local-gk}.
In particular, geometric locality ensures that the decomposition of $G_k$ contains only $\mathcal{O}(N)$ Pauli strings and that each of their coefficients is obtained with $\mathcal{O}(1)$ classical computational effort, so that the total classical cost of the decomposition scales as $\mathcal{O}(N)$.

\subsubsection{Explicit Pauli decomposition of \texorpdfstring{$G_1$}{G1} using Pauli path propagation}\label{app:impl-local-ppp}
Let us start by considering a single summand in \cref{eq:G1-recall}, corresponding to a Pauli $P_h$, and the effect of the repeated conjugations as appearing in~\cref{eq:app_one_conjug}. In particular we wish to obtain the Pauli decomposition
\begin{equation}\label{eq:dec_one_ph}
    \e^{-\im s\,\mathrm{ad}_A}\,\e^{-\im\sigma\,\mathrm{ad}_B}(P_h)
  = \sum_r b_r(s, \sigma) \,P_r.
\end{equation}
To do so, we proceed through Heisenberg evolution (called in this setting, \emph{Pauli path propagation}) of $P_{h}$, and represent the $N$ successive conjugations involved in \cref{eq:app_one_conjug} as a \emph{binary tree} of height $N$ such that: (i) each node is a Pauli string; (ii) each edge (we call it a \emph{branch}) carries a weight; and (iii) each level $\eta$ corresponds to a conjugation by the unitary $\e^{-\im\theta_{\eta}P_{\eta}}$.
The root node ($\eta = 0$ of the tree) is then the Pauli string $P_h$, and at each level $\eta>0$ a given parent node $R$ has two branches: the \emph{identity branch}, to child $R$, and the \emph{commutator branch}, to child $P_\eta R$, whose weights are set by whether $P_\eta$ commutes or anticommutes with $R$.
As already seen in \cref{eq:pauli-conj-rule}, the identity (or commutator) branch maps to a child node $R$ (or $P_\eta R$) with weight $w_0^{(\eta)}(R)$ (or $w_1^{(\eta)}(R)$) given by
\begin{equation}
  (w_0^{(\eta)}(R),\,w_1^{(\eta)}(R))
  = \begin{cases}
    (1,\,0), & [P_\eta,R]=0,\\[2pt]
    (\cos 2\theta_\eta,\,-\im\sin 2\theta_\eta), & \{P_\eta,R\}=0.
  \end{cases}
  \label{eq:exact-layer-weights}
\end{equation}

A path in such a tree is a choice of which branches are taken from the root node to one of the leaf nodes (at level $\eta  =N$). It can be labeled by a subset of integers $S \subseteq [N]:=\{1,\dots,N\}$ that identifies which commutator branches have been taken along the way. For a given $S$, we denote as $P_S^{(\eta)}$ the Pauli string corresponding to the node belonging to the path at level $\eta$ (such that $P_S^{(0)}=P_h$ for any $S$) and define $P_S:=P_S^{(N)}$ to be the Pauli string corresponding to the leaf node of the path.
For instance, $S = \emptyset$ corresponds to a path where the identity branch is systematically taken such that $P^{(\eta)}_{S}=P_S=P_h$.
Finally, define as
\begin{equation}
  g_S(s):=\prod_{\eta\notin S}w_0^{(\eta)}(P_S^{(\eta)})\prod_{\eta\in S}w_1^{(\eta)}(P_S^{(\eta)}),
  \label{eq:path-weight-0}
\end{equation}
the weight of the path. 
With these, the Pauli decomposition of \cref{eq:app_one_conjug} can be recast as
\begin{equation}
\e^{-\im s\,\mathrm{ad}_A}\,\e^{-\im\sigma\,\mathrm{ad}_B}(P_h) = \sum_S g_S(s) P_S.
\label{app:one_path_decomp}
\end{equation}

Different paths $S$ may terminate on the same leaf Pauli string; collecting them yields the coefficient of a given $P_r$ in \cref{eq:dec_one_ph}, namely $b_r(s,\sigma)=\sum_{S:\,P_S=P_r}g_S$. Since these path weights are summed with their signs, contributions to the same $P_r$ can cancel, so the resulting Pauli decomposition benefits from such cancellations and its $1$-norm is generally smaller than the path $1$-norm $\sum_S|g_S|$ obtained without this collection (as used in the non-local case of \cref{app:nongeom-sampler}).

For a given path $S$, the classical effort to obtain $P_S$ or to compute the weight $g_S(s)$ has complexity $\mathcal{O}(N)$. For the latter, we note that its expression can be obtained in closed form, as a product of $\cos(2 \theta_{\eta})$ and $\sin (2 \theta_{\eta})$ from \cref{eq:exact-layer-weights}.
Obtaining the Pauli decomposition of \cref{app:one_path_decomp} is, in general, not tractable since the tree has $2^N$ paths.
However, all the paths may not contribute to the sum. We say that a path is \emph{valid} whenever its weight is non-zero, i.e., whenever it does not include a branch with zero weight. Thus, the complexity of obtaining the decomposition of \cref{app:one_path_decomp} scales with the number of valid paths.

As was discussed in the proof of \cref{lem:rate-parameter}, in the geometrically local case, only $\mathcal{O}(1)$ of the rotations generated by the $P_{\eta}$~\eqref{eq:app_one_conjug} act non-trivially (i.e., yield valid commutator branches): only a constant number of paths are valid.
Since these $\mathcal{O}(1)$ non-trivial rotations are identified directly from the $\mathcal{O}(1)$-site support of $P_h$, each site hosting $\mathcal{O}(1)$ terms of the Hamiltonian, all the coefficients appearing in the Pauli decomposition of \cref{eq:dec_one_ph} can be obtained in $\mathcal{O}(1)$ computational effort.

The previous construction can be repeated for all the $P_h$ components of \cref{eq:G1-recall}.
Again invoking geometric locality, we are guaranteed that we have $\mathcal{O}(N)$ of such components as they arise from the commutator $[A, B]$. Integrating the coefficients obtained over $\sigma \in [0,s]$, and summing them over all the $P_h$ involved, weighted by the coefficient $c_h$, we obtain explicitly all the coefficients appearing in the Pauli decomposition of $G_1(s)$~\eqref{eq:G1-recall} in $\mathcal{O}(N)$ computational effort.
The evaluation of \cref{eq:app_hi_norm,eq:app_int_rate}, and thus the implementation of TE-PAI, follow.  

\subsubsection{Explicit Pauli decomposition for \texorpdfstring{$G_k$}{Gk}}\label{app:impl-local-gk}
The previous construction is readily generalised to higher orders of $G_k$ and arbitrary splittings of the Hamiltonian: albeit requiring more classical effort, this remains scalable. To see that, start with the kernel form~\eqref{eq:Gk-remainder-explicit}
that expresses $G_k(s)$ as a weighted sum of nested commutators $C_\nu$ \eqref{eq:Cnu-explicit} conjugated by $L_k^{j_\nu}(s,\sigma)$~\eqref{eq:Lk-two-arg}.
Given geometric locality, we still have that at a fixed, but arbitrary, order $k$: (i) each
$C_\nu$ is supported on $\mathcal{O}(1)$ sites and thus can be decomposed as $\mathcal{O}(1)$ Pauli strings with constant support, and that (ii) $\mathcal{O}(1)$ Pauli rotations, among the $\mathcal{O}(N)$ appearing in the decomposition of $L_k^{j_\nu}(s,\sigma)$, act non-trivially. 
Hence, we obtain a scaling in $\mathcal{O}(N)$ for the computational effort to obtain an explicit Pauli decomposition of $G_k(s)$ used for the evaluation of \cref{eq:app_hi_norm,eq:app_int_rate}, and the implementation of TE-PAI.
When $H$ is translation invariant, only $\mathcal{O}(1)$ of these coefficients are inequivalent, so the pre-computation reduces to $\mathcal{O}(1)$ classical effort.

\subsection{Non-local case: Pauli-path sampler}
\label{app:nongeom-sampler}

For non-local Hamiltonians the previous approach could fail: we are not guaranteed any more that only a constant number of rotations act non-trivially, such that we may have an exponentially large number of contributing Pauli strings appearing in \cref{eq:dec_one_ph} and thus in \cref{eq:G_gen}. 
To overcome this problem we will resort to an alternative h.i. decomposition of $G_k$ that lends itself to scalable implementation. 
As already discussed, any h.i. decomposition of the generator of the evolution can be used for TE-PAI, albeit the resources required will vary from one to another. 
To proceed, we start by introducing a first alternative decomposition of $G_1$ (\cref{app:path_dec}) and an alternative tree construction (\cref{app:alt_tree}).
We then combine these two to present our Pauli path sampler that enables the implementation of $\PT_1$ (\cref{app:pauli_path_sampler}). We further detail its implementation, verify its consistency and show that it does not alter the gate-count scalings obtained in \cref{app:non-geometric-gate-count}, despite the overhead introduced by the alternative decomposition (\cref{sec:pauli_path_sampler_more}). We conclude by detailing its generalisation (\cref{app:pauli_sampler_non_loc_gen}).

\subsubsection{The path decomposition}\label{app:path_dec}
Let us first extend the definitions introduced in \cref{app:impl-local-ppp}. Define an \emph{extended path} as a tuple $\mathbf{x}:=( S, h, \sigma)$ with $S$ a path, together with an index $h$ of the Pauli strings appearing in \cref{eq:G1-recall}, and a value $\sigma\in[0,s]$. 
In particular the choice of $h$ fixes the root node of the tree to be $P_h$, while the choice of $\sigma$ fixes the weights of its edges~\eqref{eq:exact-layer-weights} indirectly through the angles $\theta_{\eta}$ in \cref{app:angles_def}. 
$P^{(\eta)}_{\mathbf{x}}$ now denotes the node at level $\eta$, after $\eta$ conjugations, for the path $S$ starting at the root node $P_h$, and $P_{\mathbf{x}}:=P^{(N)}_{\mathbf{x}}$ is the Pauli string corresponding to the leaf node.
Additionally, the weight of the extended path is defined as
\begin{equation}
  g_\mathbf{x}(s):=c_h\prod_{\eta\notin S}w_0^{(\eta)}(P^{(\eta)}_{\mathbf{x}})\prod_{\eta\in S}w_1^{(\eta)}(P^{(\eta)}_{\mathbf{x}}),
  \label{eq:path-weight-1}
\end{equation}
that now accounts for the coefficient $c_h$. 

With these extended definitions, the decomposition in \cref{app:one_path_decomp} is lifted to a decomposition of $G_1$ given by
\begin{equation}
  G_1(s) = \int g_{\mathbf{x}}(s)\,P_\mathbf{x}\,\diff \mathbf{x},
  \qquad\textrm{with} \quad
  \int\diff \mathbf{x} := \sum_h\int_0^s\diff\sigma\,\sum_S.
  \label{eq:exact-path-coeff-ngf}
\end{equation}
We call such decomposition a \emph{path decomposition}, and highlight that it now involves infinitely many terms (indexed by $\mathbf{x}$) and introduces redundancies (as we can have $P_{\mathbf{x}}=P_{\mathbf{x}'}$ even when $\mathbf{x}\neq \mathbf{x}'$). Nonetheless, it is a valid h.i. decomposition since any $P_\mathbf{x}$ is a Pauli string. Hence, TE-PAI results (\cref{app:tepai-details}) equally apply. 
Still, it should be stressed that it yields a different $1$-norm, compared to the norm induced by a Pauli decomposition. In particular, denoting as $\|\cdot\|_{1, \textrm{path}}$ the norm corresponding to the path decomposition, we have that
\begin{equation}
  \|G_1(s)\|_{1,\mathrm{path}} := \int|g_\mathbf{x}(s)|\,\diff \mathbf{x} \;\ge\; \|G_1(s)\|_{1,\mathrm{Pauli}},
  \label{eq:path-norm-ngf}
\end{equation}
due to the redundancies mentioned. \cref{eq:path-norm-ngf} becomes an equality only when, for any Pauli string $P_j$, all the coefficients $\{g_{\mathbf{x}}(s) | \forall \mathbf{x}: P_{\mathbf{x}}=P_j \}$ have the same sign. 

As it is, this path decomposition cannot yet be used for implementation since it is not clear how to evaluate $\|G_1\|_{1, \textrm{path}}$ and $\lambda_1$, which are prerequisites to the implementation of TE-PAI as discussed in \cref{app:meth}. In particular, these require taking a sum of the weights $|g_{\mathbf{x}}(s)|$ over potentially as many as $2^N$ paths.
Given that these weights do not have a common structure and depend on the path taken, not known a priori as it relies on whether $P_\eta$ anticommutes with the previous node $P_{\mathbf{x}}^{(\eta)}$ as per \cref{eq:exact-layer-weights}, it seems that this exponential overhead cannot be avoided.
The alternative tree introduced in the next section avoids such a pitfall.

\subsubsection{An alternative tree}\label{app:alt_tree}
Let us start by introducing upper bounds on the rotation angles $\theta_{\eta}$~\eqref{app:angles_def}.
These are defined as 
\begin{equation}\label{eq:bar_theta}
    \bar\theta_\eta:= |\tau \omega_{\eta} c_\eta|.
\end{equation} 
By construction, these satisfy $|\theta_\eta|\le\bar\theta_\eta$,
for $0\le\sigma\le s\le\tau$. Furthermore, since $\omega_{\eta}=1$ (as $k=1$) we also have
\begin{equation}\label{eq:sum_th}
\sum_\eta\bar\theta_\eta =\tau\beta,    
\end{equation}
where we recall that
$\beta:=\|H\|_{1,\mathrm{Pauli}}$. 

Based on these upper-bound angles, we consider an alternative tree to the one defined in \cref{app:impl-local-ppp}. The only change is that the weights of the branches~\eqref{eq:exact-layer-weights} are replaced by (all the quantities corresponding to this alternative tree are highlighted with overlines):
\begin{equation}
  (\bar w_0^{(\eta)}, 
  \bar w_1^{(\eta)} ):=(1, 2\bar\theta_\eta).
  \label{eq:ub-layer-weights}
\end{equation}
In particular, note that these weights do not depend any more on any parent Pauli $R$, such as the root node $P_h$ of the tree, and are thus uniform at each level of the tree.
We further have that 
\begin{equation}\label{eq:upper_weights}
|w_b^{(\eta)}(R)|\le\bar w_b^{(\eta)},  
\end{equation}
for any $R$, since $|\cos 2\theta_\eta|\le 1$ and
$|\sin 2\theta_\eta|\le 2\bar\theta_\eta$. 
Given an extended path $\mathbf{x}$, we redefine the weight~\eqref{eq:path-weight-1} with respect to this alternative tree through
\begin{equation}
  \bar g_\mathbf{x}(s) := |c_h|\prod_{\eta\notin S}\bar w_0^{(\eta)}\prod_{\eta\in S}\bar w_1^{(\eta)}
  \;\ge\;|g_{\mathbf{x}}(s)|,
  \label{eq:dominating-coeff-ngf-bar}
\end{equation}
with the inequality following from \cref{eq:upper_weights}.
In addition, the operator $\overline{G}_1$, the dominating counterpart to $G_1$~\eqref{eq:exact-path-coeff-ngf}, is defined as
\begin{equation}
  \overline{G}_1(s) = \int \bar{g}_{\mathbf{x}}(s)\,P_\mathbf{x}\,\diff \mathbf{x}.
  \label{eq:exact-path-coeff-ngf-bar}
\end{equation}

Two notable features of this alternative tree will be used.
First, the sum of $\bar g_\mathbf{x}(s)$ over paths factorises across the levels (equivalently, over the $N$ conjugations), giving a closed-form for its norm:
\begin{equation}
  \|\overline{G}_1(s)\|_{1,\mathrm{path}}
  := \int\bar g_\mathbf{x}(s)\,\diff \mathbf{x}
  = \alpha_1\,s\prod_{\eta=1}^{N}(1+2\bar\theta_\eta)
  \;\le\; \alpha_1\,s\,\e^{2\beta\tau},
  \label{eq:dominating-norm-ngf-bar}
\end{equation}
with $\alpha_1=\sum_h |c_h|$~\eqref{eq:def_app_alpha1}, and the final
inequality following from $1+x\le\e^x$ and \cref{eq:sum_th}.

Second, one can sample an extended path $\mathbf{x}= (S, h, \sigma)$ from this alternative tree, according to the density
\begin{equation}\label{app:target_p}
    p(\mathbf{x}) = \frac{\bar g_\mathbf{x}(s)}{\int\bar g_\mathbf{x}(s) \,\diff \mathbf{x}}.
\end{equation}
This is achieved by: (i) sampling a $\sigma$ uniformly over the range $[0,s]$, (ii) sampling $h$ according to $p(h) = |c_h|/\alpha_{1}$, (iii) sampling a path $S$ along the alternative tree by taking the commutator branch at level $\eta$ with probability $p_1^{(\eta)}=\bar w_1^{(\eta)}/(\bar w_0^{(\eta)}+\bar w_1^{(\eta)})$ or the
identity branch with probability $p_0^{(\eta)}=\bar w_0^{(\eta)}/(\bar w_0^{(\eta)}+\bar w_1^{(\eta)})$.
Using independence of the different random variables, we see that the samples obtained follow the joint probability distribution,
\begin{equation}
    p(S, h, \sigma) \propto |c_h| \prod_{\eta\notin S}\bar w_0^{(\eta)}\prod_{\eta\in S}\bar w_1^{(\eta)} = \bar g_\mathbf{x}(s),
\end{equation}
which is the targeted distribution~\eqref{app:target_p}.

\subsubsection{The Pauli-path sampler}\label{app:pauli_path_sampler}
Let us now consider another decomposition of $G_1$, called the signed path decomposition, given by
\begin{equation}
  G_1(s) = \int_{\mathbf{x}} \frac{\bar g_\mathbf{x}(s)+g_\mathbf{x}(s)}{2}\,P_\mathbf{x}
  + \frac{\bar g_\mathbf{x}(s) - g_\mathbf{x}(s)}{2}\,(-P_\mathbf{x}).
  \label{eq:split-decomp-ngf-2}
\end{equation}
We introduced additional redundancies through the splitting of each of the Paulis $P_{\mathbf{x}}$ appearing in \cref{eq:exact-path-coeff-ngf} into itself and $(-P_{\mathbf{x}})$. Note that these additional $(-P_{\mathbf{x}})$ are also valid h.i. operators, such that \cref{eq:split-decomp-ngf-2} is a valid h.i. decomposition. Notably, we now have that this decomposition entails an h.i. norm (denoted $\| \cdot \|_{1,\textrm{spath}}$) that has closed form
\begin{align}\label{app:g1_spath}
\begin{split}
  \|G_1(s)\|_{1,\mathrm{spath}} &= \frac{1}{2} \int_{\mathbf{x}} |\bar g_\mathbf{x}(s)+g_\mathbf{x}(s)|
  + |\bar g_\mathbf{x}(s) - g_\mathbf{x}(s)| \\ &= \int_{\mathbf{x}} |\bar g_\mathbf{x}(s)| = \alpha_1\,s\prod_{\eta=1}^{N_A+N_B}(1+2\bar\theta_\eta)
\end{split}
\end{align}
To obtain the second equality, we used that both terms in the absolute values are non-negative from ~\cref{eq:dominating-coeff-ngf-bar}. For the last one, we used \cref{eq:dominating-norm-ngf-bar}. 
In turn, we can evaluate the integrated rate 
\begin{equation}
  \bar\lambda_1
  := \int_0^\tau\|G_1(s)\|_{1,\mathrm{spath}}\,\diff s
  = \frac{\alpha_1\,\tau^2}{2}\prod_{\eta}(1+2\bar\theta_\eta)
  \le \frac{\alpha_1\,\e^{2\beta\tau}\,\tau^2}{2},
  \label{eq:lambdabar-ngf}
\end{equation}
where we used~\cref{eq:dominating-norm-ngf-bar} for the last inequality.
Based on these, we are now in measure to describe our circuit sampler.

Here, we adapt the implementation prescription from \cref{app:tepai-impl} to the decomposition of \cref{eq:split-decomp-ngf-2}.
First draw an integer $N_{\mathrm{g}}$ from a Poisson
distribution with the expected number of events given by~\cref{eq:tepai-Ngates-derived} with integrated rate $\bar \lambda_1$ defined in \eqref{eq:lambdabar-ngf}. 
The integer $N_{\mathrm{g}}$ indicates the number of non-identity channels (i.e.\ the number of rotations) in the circuit.
For each insertion $m\in\{1, \hdots, N_{\mathrm{g}}\}$:
\begin{enumerate}\label{enum:tepai-items} 
\item Draw a time $s_m\in[0,\tau]$ from the density $p(s) = \|G_1(s)\|_{1,\mathrm{spath}}/\bar  {\lambda}_1$, with $\|G_1(s)\|_{1,\mathrm{spath}}$ defined in \cref{app:g1_spath}. 
\item Draw an extended path $\mathbf{x}_m=(h_m, S_m, \sigma_m)$ from the alternative tree distribution~\eqref{app:target_p}. From this, obtain the Pauli string $P_{\mathbf{x}_m}$ together with the weights $g({\mathbf{x}_m})$~\eqref{eq:path-weight-1} and $\bar g({\mathbf{x}_m})$~\eqref{eq:dominating-coeff-ngf-bar}.
\item Pick the sign of the Pauli $P_{\mathbf{x}_m}$ to be $a_m = \pm 1$ with probability 
\begin{equation}
  p(a_m | \mathbf{x}_m)=\frac{\bar g_{\mathbf{x}_m}(s)+ a_m g_{\mathbf{x}_m}(s)}{2\,\bar g_{\mathbf{x}_m}(s)}.
  \label{eq:p-plus-ngf}
\end{equation}
\item Define the Pauli-rotation channel $\mathcal{R}_{\mathbf{x}_m}(\theta)[\cdot]:=\e^{-\im \theta P_{\mathbf{x}_m}}[\cdot]\e^{\im \theta P_{\mathbf{x}_m}}$. Pick $\mathcal{R}_{\mathbf{x}_m}(a_m \Delta)$ with probability $p_2/(p_2+p_3)$~\eqref{eq:tepai_probability} and $\mathcal{R}_{\mathbf{x}_m}(\pi)$ otherwise, and retain the sign $\varepsilon_m=\mathrm{sgn}(\gamma_l)$ of the corresponding coefficient in~\eqref{eq:PAI-decomp}.
\end{enumerate}
Apply the $N_{\mathrm{g}}$ sampled channels in increasing order of sampled times $s_m$ to obtain one circuit instance, of overall sign $\varepsilon=\prod_{m=1}^{N_{\mathrm{g}}}\varepsilon_m$ (the path signs are already carried by the sampled rotations).
As in \cref{app:tepai-impl}, the value of the observable $O$ measured on the output state is then multiplied by the sign $\varepsilon$ and rescaled by the sampling-overhead factor $\gamma_\infty(\tau)$ of \cref{lem:tepai-app} to provide the unbiased estimate of $\langle O\rangle$.

\subsubsection{Consistency and gate-count scalings}\label{sec:pauli_path_sampler_more}

Compared to \cref{app:tepai-impl}, the previous implementation of the circuit sampler introduces a few differences. 
The sampling of $N_{\mathrm{g}}$ together with Step~1 are similar to \cref{app:tepai-impl}, up to the notation used here.
However, Step~2 of \cref{app:tepai-impl} has been split here into Step~2 and Step~3, and we wish to ensure consistency. 
On one hand, according to the implementation prescription, the choice of a given pair of path and sign, $(\mathbf{x}_m, a_m)$, is obtained with probability
\begin{equation}
    p(\mathbf{x}_m, a_m) = \frac{\bar g_\mathbf{x}(s)}{\|G_1(s)\|_{1,\mathrm{spath}}} \frac{\bar g_{\mathbf{x}_m}(s)+ a_m g_{\mathbf{x}_m}(s)}{2\,\bar g_{\mathbf{x}_m}(s)} = \frac{\bar g_{\mathbf{x}_m}(s)+ a_m g_{\mathbf{x}_m}(s)}{2\,\|G_1(s)\|_{1,\mathrm{spath}}} 
\end{equation}
where we used the probability of sampling an extended path~\eqref{app:target_p}, the probability of sampling a sign~\eqref{eq:p-plus-ngf}, and the definition of the signed path norm~\eqref{app:g1_spath}.
On the other hand, the probability of selecting the operator $a_m P_{\mathbf{x}_m}$ according to Step~2 of \cref{app:tepai-impl}, based on the h.i. decomposition of \cref{eq:split-decomp-ngf-2}, is given by
\begin{equation}
    p(a_mP_{\mathbf{x}_m}) = \frac{|(\bar g_{\mathbf{x}_m}(s)+ a_m g_{\mathbf{x}_m}(s))/2|}{\int |(\bar g_{\mathbf{x}_m}(s)+ a_m g_{\mathbf{x}_m}(s))/2|} = \frac{\bar g_{\mathbf{x}_m}(s)+ a_m g_{\mathbf{x}_m}(s)}{2\,\|G_1(s)\|_{1,\mathrm{spath}}},
\end{equation}
where we used \cref{app:g1_spath} for the second equality.
One can see that these two probabilities are the same.
Finally, Step~4 is similar to Step~3 of \cref{app:tepai-impl}, and the additional post-processing step is the same.
That is, the previous steps are consistent with an implementation of TE-PAI applied to \cref{eq:split-decomp-ngf-2}.

\medskip
Finally, given that we resort to an alternative h.i. decomposition~\eqref{eq:split-decomp-ngf-2} for the implementation, we wish to ensure that the corresponding gate-count scaling does not change. Recall that the gate count required by TE-PAI (see Lemma~1) is induced by the value of the integrated rate $\lambda$ corresponding to a particular h.i. decomposition. The bound on the rate corresponding to our circuit sampler, derived in \cref{eq:lambdabar-ngf}, matches the bound used in \cref{eq:lambda-nongeom} at $k=1$, such that the reported gate-count scaling (for this simplified case) remains the same. 

\subsubsection{Generalisation}\label{app:pauli_sampler_non_loc_gen}
The previous construction is readily generalised to higher orders of $G_k$ and arbitrary splittings of the Hamiltonian: albeit requiring more classical effort, this remains scalable and the gate-count scalings reported in \cref{app:apply-tepai} remain the same. In what follows, we review the main steps required for this generalisation.

Start with the kernel form~\eqref{eq:Gk-remainder-explicit}
that expresses $G_k(s)$ as a weighted sum of nested commutators $C_\nu$ \eqref{eq:Cnu-explicit} conjugated by $L_k^{j_\nu}(s,\sigma)$~\eqref{eq:Lk-two-arg}. Expand each nested commutator as $C_\nu=\sum_h c^{(\nu)}_h P_h$, and each
$L_k^{j_\nu}(s,\sigma)$ as a product of $N_\nu$ single-Pauli rotations $\e^{-\im\theta_\eta P_\eta}$ with angles $\theta_{\eta}$~\eqref{app:angles_def}.
To account for the several nested commutators involved (indexed by $\nu$), redefine the extended path, from \cref{app:path_dec}, as $\mathbf{x}=(\nu,h,\sigma,S)$ with paths now depending on $\nu$ and labeled by a subset $S\subseteq [N_{\nu}]:=\{1,\dots,N_\nu\}$. Its corresponding weight is defined as
\begin{equation}
  g_\mathbf{x}(s):=K_\nu(s,\sigma)\,c^{(\nu)}_h\prod_{\eta\notin S}w_0^{(\eta)}(P^{(\eta)}_{\mathbf{x}})\prod_{\eta\in S}w_1^{(\eta)}(P^{(\eta)}_{\mathbf{x}}),
  \label{eq:path-weight}
\end{equation}
where $K_\nu$ are the kernels of \cref{eq:Gk-remainder-explicit}.
In turn, the path decomposition of \cref{eq:exact-path-coeff-ngf} becomes
\begin{equation}
  G_k(s)=\int g_{\mathbf{x}}(s)\,P_\mathbf{x}\,\diff \mathbf{x},
  \qquad \textrm{with} \quad
  \int\diff \mathbf{x}:=\sum_\nu\sum_h\int_0^s\!\diff\sigma\sum_{S\subseteq[N_\nu]}.
  \label{eq:path-decomp}
\end{equation}

For the alternative tree (from \cref{app:alt_tree}): the definition of the weights of the branches~\eqref{eq:ub-layer-weights}, and the definition of the upper bound $\bar \theta_\eta$~\eqref{eq:bar_theta} on the angles $\theta_\eta$, both remain the same. 
However the Suzuki coefficients $\omega_{\eta}$ appearing in \cref{eq:bar_theta} depend on the order $k$, and we now have the bound
\begin{equation}\label{eq:bound_sum_eta_bar}
    \sum_\eta\bar\theta_\eta\le
\chi_k\beta\tau,
\end{equation} 
as seen earlier in \cref{eq:chi-k-def}. The weight~\eqref{eq:dominating-coeff-ngf-bar} of an extended path with respect to this alternative tree is now defined through
\begin{equation}
  \bar g_\mathbf{x}(s) := K_\nu(s,\sigma)\, |c^{(\nu)}_h|\prod_{\eta\notin S}\bar w_0^{(\eta)}\prod_{\eta\in S}\bar w_1^{(\eta)}.
  \label{eq:dominating-coeff-ngf-bar-gen}
\end{equation}
In turn, the $1$-norm corresponding to this alternative tree becomes
\begin{equation}
  \|\overline{G}_k(s)\|_{1,\mathrm{path}}
  := \int\bar g_\mathbf{x}(s)\,\diff \mathbf{x}
  = s^k\sum_\nu w_\nu\,\|C_\nu\|_{1,\mathrm{Pauli}}\prod_{\eta=1}^{N_\nu}(1+2\bar\theta_\eta)
  \;\le\; \tilde\alpha_k\,\e^{2\chi_k\beta\tau}\,s^k.
  \label{eq:dominating-norm-k-ngf}
\end{equation}
For the second equality, we used \cref{eq:wnu-int} and $\|C_\nu\|_{1,\mathrm{Pauli}}:=\sum_h |c_h^{(\nu)}|$. For the inequality, we used \cref{eq:bound_sum_eta_bar} and $\tilde\alpha_k = \sum_\nu w_\nu\|C_\nu\|_{1,\mathrm{Pauli}}$.
Then, integrating over $s\in[0,\tau]$ gives the
integrated rate
\begin{equation}
  \bar\lambda_k
  := \int_0^\tau\|\overline{G}_k(s)\|_{1,\mathrm{path}}\,\diff s
  \le \frac{\tilde\alpha_k\,\e^{2\chi_k\beta\tau}\,\tau^{k+1}}{k+1},
  \label{eq:lambdabar-k-ngf}
\end{equation}
Both \cref{eq:dominating-norm-k-ngf,eq:lambdabar-k-ngf} are used to generalise the implementation of \cref{app:pauli_path_sampler}.
Finally, we note that the upper bound in \cref{eq:lambdabar-k-ngf} matches the bound on $\lambda_k$ used to derive the gate-count scalings~\eqref{eq:lambda-nongeom}. This shows that despite the use of an alternative decomposition for the purpose of implementation, our gate-count scalings provided in \cref{app:apply-tepai} remain valid.

\section{Error metric and its physical interpretation}
\label{app:error-metric}

This section records the worst-case error metric behind the gate-count comparison of the main text and how we compute it.
There, an observable $O$ with $\|O\|_\infty\le1$ is estimated after evolving an input state $\rho$ with either the exact unitary $U=U(t)$ or the Trotterized circuit $V=\mathcal{S}_k(t/r)^r$.
Writing $\Phi_U(\rho):=U\rho\,U^\dagger$ for the induced channel, the bias of the estimate obeys
\begin{equation}
    \bigl|\tr[O\,(\Phi_U(\rho)-\Phi_V(\rho))]\bigr| \;\le\; \|\Phi_U-\Phi_V\|_\diamond,
    \label{eq:em-bias-diamond}
\end{equation}
where there exists a state and an observable that achieve the equality~\cite{watrous2018theory}.
The diamond-norm distance is therefore precisely the worst-case bias entering the RMSE criterion Eq.~(7).

For unitary channels the diamond-norm distance is determined by the spectrum of $W:=U^\dagger V$~\cite{aharonov1998quantum,watrous2018theory,haah2023query}.
Let $2\alpha$ be the arc length of the smallest arc of the unit circle containing the eigenvalues of $W$; for the Trotter error considered here all eigenvalues lie close to $1$, so $\alpha\le\pi/2$ holds comfortably.
Then~\cite{aharonov1998quantum,watrous2018theory}
\begin{equation}
    \min_{\theta\in\mathbb{R}}\|U-\e^{\im\theta}V\|_\infty = 2\sin(\alpha/2),
    \qquad
    \tfrac{1}{2}\|\Phi_U-\Phi_V\|_\diamond = \sin\alpha.
    \label{eq:em-identity}
\end{equation}
This gives the relation between the channel distance and the phase-optimised spectral norm $m:=\min_\theta\|U-\e^{\im\theta}V\|_\infty$,
\begin{equation}
    \|\Phi_U-\Phi_V\|_\diamond = 2m\sqrt{1-m^{2}/4}.
    \label{eq:em-dch-from-opnorm}
\end{equation}
We calculate the diamond norm through this optimisation over $\theta$: the spectrum of $W$ is obtained by exact diagonalisation, which is what restricts the worst-case gate-count study of the main text to $n\le14$.

\end{document}